\documentclass[acmsmall,screen,nonacm]{acmart}

\usepackage{booktabs}
\usepackage{subcaption} 
\usepackage{amsmath}
\usepackage{eurosym}
\usepackage{wrapfig}
\usepackage{color}
\usepackage{tabularx}
\usepackage{booktabs}
\usepackage{array}
\usepackage{ifthen}
\usepackage{stmaryrd}
\usepackage{mathrsfs}
\usepackage{proof}
\usepackage{multirow}
\usepackage{mathpartir}
\usepackage[noend]{algpseudocode}
\usepackage{caption}
\usepackage{enumitem}
\usepackage{tikz}
\usepackage{pgfplots}
\usetikzlibrary{pgfplots.groupplots}
\pgfplotsset{width=7.5cm,compat=1.12}
\usepgfplotslibrary{fillbetween}
\usepackage{fontawesome}
\usepackage{ulem}
\newtheorem{thm}{Theorem}
\newtheorem{cor}{Corollary}
\newtheorem{lma}{Lemma}

\newcommand{\DOT}{\top}

\newcommand{\A}{\mathcal{A}}

\newcommand{\D}{\Sigma}

\newcommand{\den}[2][{}]{[\hspace{-2.5pt}[ #2 ]\hspace{-2.5pt}]_{#1}}

\newcommand{\DERS}[1]{\mathrel{\raisebox{-.1em}{\scriptsize$\xrightarrow{\begin{array}{@{}c@{}}#1\\[-.4em]\end{array}}$}}}

\newcommand{\rand}{\mathbin{\texttt{\&}}}
\newcommand{\alt}{\mathbin{\texttt{|}}}
\newcommand{\rnot}{\texttt{\char`~}}

\newcommand{\COMPL}[1]{\complement(#1)}

\newcommand{\LTP}[2][{}]{{\ifthenelse{\equal{#2}{}}{\mathfrak{T}_{#1}}{\mathfrak{T}_{#1}(#2)}}}

\newcommand{\eqdef}{\stackrel{\textsc{\raisebox{-.15em}{\tiny{def}}}}{=}}

\newcommand{\BOOL}{\mathbb{B}}
\newcommand{\TT}{\mathbf{true}}
\newcommand{\FF}{\mathbf{false}}
\newcommand{\band}{\mathrel{\mathbf{and}}}
\newcommand{\bor}{\mathrel{\mathbf{or}}}
\newcommand{\bif}{\mathbf{if}}

\newcommand{\botherwise}{\mathbf{otherwise}}
\newcommand{\biff}{\boldsymbol{\Leftrightarrow}}

\newcommand{\bnot}{\mathbf{not}}

\newcommand{\eps}{\texttt{()}}

\newcommand{\emp}{\bot}

\newcommand{\DERName}{\partial}
\newcommand{\DER}[3][{}]{\ifthenelse{\equal{#1}{}}{\DERName_{#2}(#3)}{\DERName_{#2}^{#1}(#3)}}

\newcommand{\RE}{\mathcal{R}}
\newcommand{\ERE}{\mathcal{R}}

\newcommand{\rden}[1]{{\ifthenelse{\equal{#1}{}}{{\mathbf{L}}}{\mathbf{L}(#1)}}}

\newcommand{\IFF}{\Leftrightarrow}

\newcommand{\caret}{\char`^}

\newcommand{\bslash}[1]{\texttt{\char`\\#1}}

\newcommand{\tuple}[1]{\langle{#1}\rangle}
\newcommand{\pair}[2]{\tuple{#1,#2}}

\newcommand{\dollar}{\char`$}

\newcommand{\IsNullableName}{\mathit{Null}}

\newcommand{\AlwaysNullableName}{\mathit{Null}_{\forall}}

\newcommand{\IsNullable}[2][{}]{\IsNullableName_{#1}(#2)}

\newcommand{\AlwaysNullable}[1]{\AlwaysNullableName(#1)}

\newcommand{\str}[1]{\textrm{"}\texttt{#1}\textrm{"}}

\newcommand{\ReverseOf}[1]{#1^r\!}

\newcommand{\ACC}[3][{}]{{\ifthenelse{\equal{#3}{}}{{\mathbb{L}_{#2}^{#1}}}{\mathbb{L}_{#2}^{#1}(#3)}}}

\newcommand{\deriv}[2][{}]{{\ifthenelse{\equal{#2}{}}{{D_{#1}}}{D_{#1}(#2)}}}

\newcommand{\IsLazy}[1]{{\ifthenelse{\equal{#1}{}}{{\textit{IsLazy}}}{\textit{IsLazy}(#1)}}}

\newcommand{\REV}[1]{\ReverseOf{#1}}

\newcommand{\SAT}[2][{}]{\ifthenelse{\equal{#1}{}}{\textbf{SAT}(#2)}{\textbf{SAT}_{#1}(#2)}}

\newcommand{\la}[1]{\texttt{(?=}#1\texttt{)}}
\newcommand{\laneg}[1]{\texttt{(?!}#1\texttt{)}}
\newcommand{\lb}[1]{\texttt{(?<=}#1\texttt{)}}
\newcommand{\lbneg}[1]{\texttt{(?<!}#1\texttt{)}}

\newcommand{\RELoop}[4][{}]{#2\texttt{\{}#3{,#4}\ifthenelse{\equal{#1}{}}{}{,#1}\texttt{\}}}

\newcommand{\st}{\texttt{*}}
\newcommand{\lazy}{\texttt{?}}
\newcommand{\plus}{\texttt{+}}

\newcommand{\decr}[1]{#1{\stackrel{_\centerdot}{-}}1}

\newcommand{\CCpred}[2][{}]{\ifthenelse{\equal{#1}{}}{\psi_{\texttt{#2}}}{\psi^{#1}_{\texttt{#2}}}}

\newcommand{\ANCH}{\textrm{\scriptsize\faAnchor}}

\newcommand{\first}[1]{\pi_1(#1)}
\newcommand{\second}[1]{\pi_2(#1)}

\newcommand{\IsMatchName}{\textit{IsMatch}}
\newcommand{\IsMatch}[2]{\IsMatchName(#1,#2)}
\newcommand{\FindMatchEndName}{\textit{MatchEnd}}
\newcommand{\FindMatchEnd}[2]{\FindMatchEndName{}(#1,#2)}
\newcommand{\NoMatch}{\lightning}

\newcommand{\loc}[2]{#1\langle #2 \rangle}

\newcommand{\FindAllMatches}[2]{({#1}{\DERS{#2}})}

\newcommand{\FindMatchName}{\textrm{$\mathit{llMatch}$}}
\newcommand{\FindMatch}[2]{\FindMatchName{}(#1,#2)}

\newcommand{\blet}{\textbf{let}}
\newcommand{\bin}{\textbf{in}}

\newcommand{\NullF}[2]{\textit{Null}^{\NoMatch}_{#1}\!(#2)}
\newcommand{\NullE}[2]{\textit{Null}^{\emptyset}_{#1}(#2)}

\newcommand{\FME}[3][{}]{\ifthenelse{\equal{#1}{}}{\FindMatchEndName(#2,#3)}{#1(#2,#3)}}

\newcommand{\LANGREName}{\mathcal{M}}
\newcommand{\LANGRE}[1]{\LANGREName(#1)}

\newcommand{\IsFinal}[1]{\textit{Final}(#1)}
\newcommand{\IsInitial}[1]{\textit{Initial}(#1)}

\newcommand{\ITEBrace}[3]{\left\{\begin{array}{@{}l@{\;}l} #2,&\bif\, #1; \\ #3,&\botherwise. \end{array} \right.}

\newcommand{\myparagraph}[1]{\paragraph*{#1}}

\newcolumntype{R}{>{$}r<{$}}
\newcolumntype{L}{>{$}l<{$}}
\newcolumntype{Z}{>{$}X<{$}}

\newcommand{\MatchUniverse}{\mathcal{U}}
\newcommand{\MatchAlgebra}{\mathfrak{M}}

\newcommand{\eqIff}{\stackrel{\textsc{(ih)}}{\biff}}

\newcommand{\WEBAPP}{\footnote{\url{https://ieviev.github.io/sbre/}}}

\newcommand{\DOTpred}{\CCpred{\!\large{.}\!}}

\begin{document}

\title{Derivative Based Extended Regular Expression Matching Supporting Intersection, Complement and Lookarounds} %

\author{Ian Erik Varatalu}
\orcid{0000-0003-1267-2712}
\affiliation{
  \institution{Tallinn University of Technology}
  \country{Estonia}
}
\email{ian.varatalu@taltech.ee}

\author{Margus Veanes}
\orcid{0009-0008-8427-7977}
\affiliation{
	\institution{Microsoft}
	\country{USA}
}
\email{margus@microsoft.com}

\author{Juhan-Peep Ernits}
\orcid{0000-0002-4591-0425}
\affiliation{
	\institution{Tallinn University of Technology}
	\country{Estonia}
}
\email{juhan.ernits@taltech.ee}

\begin{abstract}
  Regular expressions are widely used in software. Various regular
  expression engines support different combinations of extensions to
  classical regular constructs such as Kleene star, concatenation,
  nondeterministic choice (union in terms of match semantics). The
  extensions include e.g. anchors, lookarounds, counters,
  backreferences. The properties of combinations of such extensions have
  been subject of active recent research.

  In the current paper we present a symbolic derivatives based
  approach to finding matches to regular expressions that, in
  addition to the classical regular constructs, also support complement,
  intersection and lookarounds (both negative and positive lookaheads
  and lookbacks). The theory of computing symbolic derivatives and
  determining nullability given an input string is presented that
  shows that such a combination of extensions yields a match semantics
  that corresponds to an effective Boolean algebra, which in turn
  opens up possibilities of applying various Boolean logic rewrite
  rules to optimize the search for matches.
  
  In addition to the theoretical framework we present an
  implementation of the combination of extensions to demonstrate the
  efficacy of the approach accompanied with practical examples.
\end{abstract}

\maketitle

\pagestyle{plain}

\section{Introduction}
\label{sec:intro}

Regular expressions are supported by standard libraries of all major
programming languages and software development kits. They are
essential in string manipulation, membership testing, text extraction
tasks on strings, etc.  In recent years symbolic derivatives based
regular expression matching has seen rapid development as being
performance wise more predictable than backtracking engines and
modern optimized regular expression engines such as e.g. Hyperscan
\cite{HyperscanUsenix19} and RE2 \cite{re2}, that are all, in essence,
relying on algorithms by \cite{Glu61} or \cite{Thom68}.

The extended regular expressions supported by backtracking engines
often involve support for combinations of extensions, e.g.
backreferences, lookarounds and balancing groups, yielding a family of
languages that has been shown by \cite{DBLP:conf/lata/CarleN09} not to
be closed under intersection. Recent work by
\cite{DBLP:journals/ieicetd/ChidaT23} has explored the expressive
power of the combination of regular expressions with backreferences
and positive lookaheads and devised a new flavor of memory automata to
express their behavior.

We present a derivative-based symbolic extended regular expression
engine, which supports \emph{intersection} (\texttt{\&}),
\emph{complement} ($\sim$) and \emph{lookaround} (\texttt{(?=)},
\texttt{(?!)}, \texttt{(?<=)} and \texttt{(?<!)})  operations,
i.e. both negative and positive lookahead and lookback operations. To
remind the reader, by example of a positive lookahead, the regex
\texttt{a(?=c)} will match \texttt{a} in the string ``\texttt{ac}''
but nothing in ``\texttt{ab}'', i.e. lookarounds do not match text but
positions in strings in a more general way than anchors (e.g. word
boundary \texttt{\bslash{}b}) \cite{friedl2006mastering}. We show that
the extended engine can be used to solve problems that are currently
unfeasible to solve with existing regular expression engines.  We also
show that all context dependent regex anchors can be expressed using
lookarounds.  We have implemented the extended engine as a library for
the .NET platform, which extends the existing .NET nonbacktracking
engine \cite{PLDI2023}, but without preserving backtracking semantics.

We demonstrate the usefulness of the extended engine by examples of precise 
text extraction in a way that is to our knowledge not possible with standard
regexes or involves a factorial blowup of the regex pattern. 
We make the case that intersection and complement are useful
tools in concise definition of regular expressions, showing examples of incremental
specification of regexes by use of complement, and exclusion of unwanted
matches by use of negation.

Our work is motivated by the fact that the current state-of-the-art symbolic regex engine
\cite{PLDI2023} does not support intersection because its semantics is difficult to combine
with all the other features in a meaningful way while maintaining backtracking semantics.
It is also a problem of backwards compatibility to, for example, support \texttt{\&} as a new operator.

\subsection*{Contributions}
\label{sec:contributions}

This paper makes the following contributions:
\begin{itemize}
      \item We build on the .NET 7 nonbacktracking
      regular expression engine \cite{PLDI2023} to develop a theory
      for symbolic derivatives based regular expression matching
      that, in addition to \emph{Kleene star}, \emph{bounded loops}
      (\emph{counters}), \emph{alternation} (which in our case is
      synonymous to \emph{union}) and \emph{concatenation}, supports
      \emph{intersection}, \emph{complement} and \emph{both positive
        and negative lookaround operations} in regexes. The result is
      a significant increment over \cite{PLDI2023} in
      terms of Theorem \ref{thm:DERS}, Theorem \ref{thm:REV} and
      Theorem \ref{thm:Match} accompanied by proofs.
    \item The theoretical results are accompanied by an implementation \cite{sbre_zenodo}
      of an extended regular expression engine that supports all
      Boolean operations on regexes, including intersection and
      complement using a natural syntax making it convenient to
      use.
    \item Due to the match semantics of such a regular expression
      language corresponding to an effective Boolean algebra, it is
      now possible to define and apply numerous Boolean rewrites and
      optimizations that yield better runtime performance. It is thus
      also possible to reason about the validity of such optimizations
      in terms of Boolean algebra.
    \item We show that the extended engine can be used to solve
      problems that are currently either impossible or infeasible to
      solve with standard regular expressions.
    \item We show that all regex anchors can be expressed using
      lookarounds as conjectured in \cite{PLDI2023}. We prove the
      conjecture together with providing an implementation as
      mentioned above.

\end{itemize}

Before proceeding to introducing the theory, let us look at some
practical examples that illustrate how the proposed combination extensions can
be used to establish the existence of matches but also locate them in
strings.

\section{Motivating examples}
\label{sec:examples}

In this section we explain the intuition behind the intersection (conjunction) and complement (negation) operators
in regexes and give some motivating examples of how they can be useful. The examples are
also available in the accompanying web application \WEBAPP{}. The following examples are 
written in the syntax of the .NET regex engine, with the addition of \texttt{\&}, \texttt{$\sim$}, and the symbol $\top$, to denote
a predicate that is true on any character. The syntax is explained in more detail in Section \ref{sec:preliminaries}.

\begin{example}[Password extraction]
\label{ex:password-extraction}

Consider the following regex
\begin{verbatim}
    ^(?=.*[a-z])(?=.*[A-Z])(?=.*\d)[a-zA-Z\d]{8,}$
\end{verbatim}
that originates from a StackOverflow question \footnote{\url{http://stackoverflow.com/questions/19605150/regex-for-password-must-contain-at-least-eight-characters-at-least-one-number-a}} 
asking for a regex for validating a password. The regex in principle is an intersection of
three lookaheads, each of which checks for 
a certain condition and a loop that checks for the length of the password.
\begin{itemize}
    \item \verb!(?=.*[a-z])! checks for at least one lowercase letter
    \item \verb!(?=.*[A-Z])! checks for at least one uppercase letter
    \item \verb!(?=.*\d)! checks for at least one digit
    \item \verb![a-zA-Z\d]{8,}! checks for at least 8 characters
\end{itemize} 
\hfill

The regex is a good example of
emulating conjunctive conditions using lookaheads for checking if there exists a match. However, 
the regex is not suitable for extracting a password from an arbitrary string, 
even when removing the anchors. 
The reason is that the regex engine will try
to match the lookaheads independently, and then backtrack to the beginning of the match
to try the next lookahead. This is an example of a limitation to
what can be expressed in traditional regular expressions. And note that due to the
lack of support for lookarounds in RE2~\cite{Cox10}, Hyperscan~\cite{HyperscanUsenix19}
and the nonbacktracking engine in {.NET} to name a few optimized open source regex libraries,
it cannot be expressed concisely using an engine that does not support lookaheads.

Now consider the following regular expression, composed of intersections which we denote
by the \texttt{\&} symbol:

\begin{verbatim}
    .*[a-z].*&.*[A-Z].*&.*\d.*&[a-zA-Z\d]{8,}
\end{verbatim}
The regular expression has a number of differences from the previous one.

\begin{itemize}
    \item The proposed regex engine will check for all four conditions at the same time,
    without any backtracking. 
    \item Upon meeting one of the conditions, the rest of the alternation is matched
    with one less condition to check. For example, after matching the lowercase letter \texttt{a} at 
    the beginning, the set \texttt{[a-z]} is satisfied, and the regex engine
    will match the rest of the input with the regex \texttt{.*[A-Z].*\&.*{\textbackslash}d.*\&[a-zA-Z{\textbackslash}d]\{7,\}},
    by effectively removing the first condition from the regex, as the remaining \texttt{.*}
    will be subsumed by the 
    other intersections. In Sections \ref{sec:derivatives} and \ref{sec:rewrites} we show how it is achieved.
    \item Because all checks are performed together at the same position, 
    intersections allow precise text extraction in a single pass, which is also 
    highly amenable to parallelization in the form of vectorization. Such optimizations 
    could be beneficial
    for tasks such as scanning for potential credential leaks.
    \item The pattern does not currently check for the validity of the entire string,
    but it can be done by appending the intersection \texttt{\caret{}.*\$} to the regex, which
    uses lookarounds to constrain the regex to match the entire line. 
    Similarly, the regex could be 
    constrained to match text between word borders by appending the intersection 
    \texttt{\bslash{}b\bslash{}S*\bslash{}b} to the regex.
\end{itemize}

Now consider modifying the regular expression to find potential password substrings 
while excluding certain ones, such as those having 2 consecutive digits 
or the word "password" in it. Such conditions cannot be checked effectively
with positive lookaheads but can be achieved with negative ones.
For example, the following regular expression uses negative lookaheads to match a password
that does not have 2 consecutive digits:

\begin{verbatim}
    .*[a-z].*&.*[A-Z].*&.*\d.*&(?!\S*\d\d)[a-zA-Z\d]{8,}
\end{verbatim}

An important detail to point out here, is that in such regular expressions, the
ranges of the lookaheads and the loop are not necessarily the same.
For example, if there was something like \texttt{"@11"} on the same line, after 
the occurrence of the potential match - it would be falsely considered a non-match, 
as \texttt{\bslash{}S*} travels further than \texttt{[a-zA-Z\bslash{}d]\{8,\}}, while the loop would
stop at \texttt{@}.

A better way of writing the pattern would be to convert the negative 
lookahead \texttt{(?!\bslash{}S*\bslash{}d\bslash{}d)} to
\texttt{(?![a-zA-Z\bslash{}d]*\bslash{}d\bslash{}d)}, that is, adjust the loop in the negative lookahead to 
the entire following regular expression body, to guarantee that
the constraints apply to the same range. 
This is where the use of complement (`\texttt{\textasciitilde}`) can help. 

\begin{verbatim}
    .*[a-z].*&.*[A-Z].*&.*\d.*&[a-zA-Z\d]{8,}&~(.*\d\d.*)
\end{verbatim}

The regex above is equivalent to the previous one, but it is constrained to \emph{the 
exact range of the potential match}, as \texttt{$\sim$(.*\bslash{}d\bslash{}d.*)} remains nullable
exactly until two sequential digits are found, after which it turns the entire intersection
of regexes into $\bot$. The negation operator is explained in detail in 
Section \ref{sec:derivatives}.

\end{example}

\begin{example}[Incremental specification of regexes through composition]
\label{ex:incremental-specification}

Consider the following regex:

\begin{verbatim}
    .*A.*&.*B.*&.*C.*        
\end{verbatim}

The regex matches any line that contains the substrings \texttt{A}, \texttt{B} and \texttt{C}
in any order.

Now consider creating an equivalent regular expression without using intersections. The result would be similar to the following:
\begin{verbatim}
    .*A.*B.*C.*|.*A.*C.*B.*|.*B.*A.*C.*|.*B.*C.*A.*|.*C.*A.*B.*|.*C.*B.*A.*
\end{verbatim}

The regex is equivalent to the previous one, but it is much more verbose, as it requires
enumerating all possible permutations of the substrings as separate alternatives.
In such scenarios the use of intersection
(\texttt{\&}) becomes very6 helpful. The number of alternations required 
is equal to the factorial of the number of substrings.

Matching substrings \texttt{A,B,C,D,E,F} this way would require 720 alternations,
while the equivalent regex with intersections
would only require 6 intersections. Note that positive lookaheads can be used for such examples
strictly as long as it is possible to
define a clear lookup range for all the substrings, which is not always the case.

Another important advantage of intersections is that they are not interlinked with each other,
and can be added or removed independently. For example, if we wanted to add a fourth 
substring \texttt{D},
we would only need to add one more intersection, instead of adding 18 alternations and 
modifying the existing 6.
And if we wanted to remove the substring \texttt{C}, we would only need to remove 
one intersection, instead of reducing and modifying all alternations down to 4.

Similarly, if we wanted exclude strings containing the substring \texttt{D},
we would only need append the negated intersection \texttt{$\sim$(.*D.*)},
instead of modifying all existing alternations to 
exclude the substring \texttt{D}.

Such incremental aspect can be significant when dealing with large regexes, 
and can be useful for tasks such as
machine learning based text extraction with regular expressions, where the regex is 
incrementally built by adding
and removing positive or negative intersection terms. Such approach also improves the 
readability of regular expressions,
as the regexes are built in a modular way.

\end{example}

\begin{example}[Matching paragraphs of text containing multiple substrings]
\label{ex:paragraph-matching}

Another practical use case for the new features is matching paragraphs
of text that meet certain conditions, as both intersection and negation
introduce new ways of matching ranges of text with complex
conditions. While it is possible to envisage sequential use of
e.g. \texttt{grep} for positively and negatively matching lines of
text on the input string when each paragraph is on a separate line, we
generalize the example to paragraphs spanning over multiple lines and
paragraphs being separated by double newlines, as is often done in
texts, e.g. in Project Gutenberg books in plain text format. The example
generalizes to any substring start and end condition and provides a
mechanism to also extract the match in addition to checking for the
existence of it.

Note that, for the sake of improved readability, we use the $\top$
predicate to denote the set of all characters, which can be thought of
as a canonical representation of \texttt{[\bslash{}s\bslash{}S]} or \texttt{[\bslash{}w\bslash{}W]}.

For example, consider the following regular expression:

\addvspace{0.25em}
\texttt{\bslash{}n\bslash{}n$\sim$($\top$*\bslash{}n\bslash{}n$\top$*)},

\addvspace{0.25em}
which effectively splits the text 
into ranges not containing 
two sequential newlines, as it becomes nullable with two newlines, and remains nullable until
two newlines are found again, after which it becomes $\bot$. As negation does not match 
at the position
of the final newline, an additional \texttt{\bslash{}n} can be appended to include that as well.

After defining the previous construct, we can extend it with intersections to match
paragraphs containing certain substrings. For example, the following regex matches
paragraphs containing the substrings \texttt{"King"}, \texttt{"Paris"} and \texttt{"English"} in any order: 

\addvspace{0.25em}
\texttt{\bslash{}n\bslash{}n$\sim$($\top$*\bslash{}n\bslash{}n$\top$*)}\&\texttt{$\top\st$King$\top\st$}\&\texttt{$\top\st$Paris$\top\st$}\&\texttt{$\top\st$English$\top\st$}

\addvspace{0.25em}
A key advantage of the approach is that it is very easy to infer the starting predicate of such
a regex and skip to it, we have all the necessary contextual information to do so. A regex that starts with $\top\st$ essentially 
means that you can skip taking derivatives until the starting predicate of the tail of the regex matches. 
In the case of \texttt{$\top\st$King$\top\st$}, the only valid predicate that changes the state of this
regex is $\CCpred{$K$}$, after which the regex becomes \texttt{ing$\top\st$|$\top\st$King$\top\st$}, 
and the next predicate always changes state, as the alternation \texttt{ing$\top\st$} either continues matching in the case of
$\CCpred{$i$}$, or turns to $\bot$, making it an "unskippable" state.

Using such information from all intersections, we can skip taking derivatives until we reach a starting predicate
that changes the state of the regex, which allows us to perform a vectorized string search procedure on any 
skippable state of the regex. Such starting predicate lookup and skipping is explained in more detail in Section~\ref{sec:startset-lookup}.

Another feature, that is exclusively available in the proposed regex engine, is the ability to
define a matching range based on non-matching constraints. For example 
matching a range of text starting with \texttt{"King"}
and ending with \texttt{"Paris"}, that must not
contain two sequential digits between each other can be concisely expressed 
with the following regex:
\texttt{King$\sim$($\top$*\bslash{}d\bslash{}d$\top$*)Paris}
which finds a match in \texttt{"The King in Paris"}, but not \texttt{"The King 11 Paris"}.
Using a negative lookahead for regexes such as this, for example 
\texttt{King(?!$\top\st$?\bslash{}d\bslash{}d)$\top\st$?Paris}, is very difficult,
as the $\top\st$ inside the lookahead will leak far past the occurrence of \texttt{Paris},
making it incorrect if any two sequential digits occur inside the text after the match.

\end{example}

\begin{example}[Conditional matching]
  \label{ex:conditional matching}

Extending the previous example, let us try to locate a paragraph
containing the word ``\texttt{charity}'' only in the case the
paragraph does not contain the word ``\texttt{honor}''. Such properties
can be represented in terms of Boolean combinations of the implication operation, which we, for now,
represent in terms of complement, union and intersection.

Let us first try to find all paragraphs that match the implication $\mathtt{charity}\implies\mathtt{honor}$.
When writing \verb!~(!$\top\st$\verb!\n\n!$\top\st$\verb!)&(~(!$\top\st$\texttt{charity}$\top\st$\verb!)|(!$\top\st$\texttt{honor}$\top\st$\verb!))!
the result is not yet precise enough, because the term
\verb!~(!$\top\st$\verb!\n\n!$\top\st$\verb!)! stipulates that the
match should not contain two consecutive newlines, but there is no
condition that tells that the string \emph{has} to have
double newlines before and after it. In the Example \ref{ex:paragraph-matching} it
worked because we did not use negative conditions, i.e. where some word should not be contained in the match. Here we need to extend the term
for matching paragraphs by appropriate lookarounds as follows:

\addvspace{0.25em}
\verb!(?<=\n\n|\A)~(!$\top\st$\verb!\n\n!$\top\st$\verb!)(?=\n\n|\Z)&(~(!$\top\st$\texttt{charity}$\top\st$\verb!)|(!$\top\st$\texttt{honor}$\top\st$\verb!))!

\addvspace{0.25em}
The result now matches all paragraphs where there is no word ``\texttt{charity}'' or there exists a word ``\texttt{honor}'' (according to the definition of the standard Boolean implication operation). To locate the paragraph where ``\texttt{charity}'' is present, but ``\texttt{honor}'' is not, we need to take the complement of the implication and write

\addvspace{0.25em}
\verb!(?<=\n\n|\A)~(!$\top\st$\verb!\n\n!$\top\st$\verb!)(?=\n\n|\Z)&~(~(!$\top\st$\texttt{charity}$\top\st$\verb!)|(!$\top\st$\texttt{honor}$\top\st$\verb!))!

\addvspace{0.25em}
\noindent{which in turn can be rewritten by applying de Morgan's rule into}

\addvspace{0.25em}
\verb!(?<=\n\n|\A)~(!$\top\st$\verb!\n\n!$\top\st$\verb!)(?=\n\n|\Z)&(!$\top\st$\texttt{charity}$\top\st$\verb!)&~(!$\top\st$\texttt{honor}$\top\st$\verb!)! 

\addvspace{0.25em}
\noindent{The resulting regular expression matches the paragraph that contains the word ``\texttt{charity}'' but does not contain ``\texttt{honor}''.}

\addvspace{0.25em}
It is possible to utilize the If-Then-Else conditionals
\verb!(?(?=!$R_1$\verb!)(!$R_2$\verb!)|(!$R_3$\verb!))! in traditional
regex engines supporting the construct. The explanation of
If-Then-Else is nontrivial \cite{friedl2006mastering}: $R_1$ acts as a
test and if the test yields true then $R_2$ gets applied, otherwise
$R_3$. $R_1$ can either be a reference to a capturing group and test
if \emph{the group participated in a match} (which is different from
the actual match) or be a lookaround and yield a Boolean result at a
particular location. As mentioned before, there are differences in the
extent of the string in which lookarounds get evaluated and the extent
of loops, as was shown in the Example \ref{ex:password-extraction}. As
a result, writing a conditional regular expression in terms of
If-Then-Else corresponding to the above requirements is error prone
and requires a great deal of care. We argue that the use of complement
and intersection (and consequently a whole Boolean algebra) together
with lookarounds simplifies the task significantly.

\end{example}

We now proceed to establishing the theory for regular expressions
extended with lookarounds, complement and intersection.

\section{Preliminaries}
\label{sec:preliminaries}

Here we introduce the notation and main concepts used in the paper.
The general meta-notation, notation for denoting strings and locations
is based on the approach taken in \cite{PLDI2023}.

As a general meta-notation throughout the paper
we write $\textit{lhs}\eqdef\textit{rhs}$
to let \textit{lhs} be \emph{equal by definition} to \textit{rhs}.
Let $\BOOL = \{\FF,\TT\}$ denote Boolean values.
Let $\D$ be a domain of
\emph{characters}.
$\pair{x}{y}$ stands for pairs of elements and let
$\first{\pair{x}{y}}\eqdef x$ and
$\second{\pair{x}{y}}\eqdef y$.

\myparagraph{Strings.}
Let $\epsilon$ or $\str{}$ denote the empty string and let $\D^*$
denote the set of all strings over $\D$. Let $s\in\D^*$.  The length
of $s$ is denoted by $|s|$.  Individual characters and strings of length 1
are not distinguished. Let $i$ and $l$ be nonnegative
integers such that $i+l \leq |s|$.  Then $s_{i,l}$ denotes the
substring of $s$ that starts from index $i$ and has length $l$, where
the first character has index 0.  In particular $s_{i,0}=\epsilon$.
For $0\leq i<|s|$ let $s_i \eqdef s_{i,1}$.
Let also $s_{-1}=s_{|s|}\eqdef\epsilon$.
E.g.,~$\str{abcde}_{1,3}=\str{bcd}$ and $\str{abcde}_{5,0}=\epsilon$.
$\REV{s}$ denotes the \emph{reverse} of $s$,
so that $\REV{s}_i = s_{|s|-1-i}$ for $0\leq i < |s|$.

\myparagraph{Locations.}
Let $s$ be a string.  A \emph{location in $s$} is a pair
$\pair{s}{i}$, where $-1\leq i \leq |s|$.  We use
$\loc{s}{i}\eqdef\pair{s}{i}$ as a dedicated notation for locations,
where $s$ is the \emph{string} and $i$ the \emph{position} of the
location.  Since $\loc{s}{i}$ is a pair, note also that
$\first{\loc{s}{i}} = s$ and $\second{\loc{s}{i}} = i$.
If $x$ and $y$ are locations, then $x<y$ iff $\second{x}<\second{y}$.
A location $\loc{s}{i}$ is \emph{valid} if $0\leq i \leq |s|$.
A location $\loc{s}{i}$ is called \emph{final} if $i=|s|$
and \emph{initial} if $i=0$.
Let $\IsFinal{\loc{s}{i}} \eqdef i=|s|$ and $\IsInitial{\loc{s}{i}} \eqdef i=0$.
We let $\NoMatch\eqdef\loc{\epsilon}{-1}$
that is going to be used to represent \emph{match failure}
and in general $\loc{s}{-1}$ is used as a \emph{pre-initial} location.
The \emph{reverse} $\REV{\loc{s}{i}}$ of a valid location $\loc{s}{i}$
in $s$ is the valid location $\loc{\REV{s}}{|s|{-}i}$ in
$\REV{s}$. For example, the reverse of the final location in $s$ is the
initial location in $\REV{s}$.
When working with sets $S$ of locations over the same string we let
$\max(S)$ ($\min(S)$) denote the maximum (minimum) location in the set 
according to the location order above.
In this context we also let $\max(\emptyset)=\min(\emptyset)\eqdef\NoMatch$
and $\REV{\NoMatch}\eqdef\NoMatch$.

Valid locations in a string $s$ can be illustrated by

\[
\begin{array}{ccccccccc} %
\loc{s}{\!0\!} & $\mbox{\Huge ${s_0}$}$& \loc{s}{\!1\!} & $\mbox{\Huge ${s_1}$}$ & \loc{s}{\!2\!} & \cdots & \loc{s}{\!|s|\!{-}\!1\!} & $\mbox{\Huge ${s_{|s|\!-\!1}}$}$ & \loc{s}{\!|s|\!} \\
\end{array}
\]
and should be interpreted as \emph{border positions}
rather than character positions.
For example, $\epsilon$ has only one valid
location $\loc{\epsilon}{0}$ that is both initial and final.

\myparagraph{Boolean algebras as alphabet theories.}
The tuple
$
\A=(\D, \Psi,
\den{\_}, \bot, \top, \vee, \wedge, \neg)
$
is called an \emph{effective Boolean algebra over $\D$} where $\Psi$ is a set
of \emph{predicates} that is closed under the Boolean connectives;
$\den{\_} : \Psi \rightarrow 2^{\D}$ is a \emph{denotation function};
$\bot, \top \in \Psi$; $\den{\bot} =
\emptyset$, $\den{\top} = \D$, and for all $\varphi, \psi \in \Psi$,
$\den{\varphi \vee \psi} = \den{\varphi} \cup \den{\psi}$,
$\den{\varphi \wedge \psi} = \den{\varphi} \cap \den{\psi}$, and
$\den{\neg \varphi} = \D \setminus \den{\varphi}$. 
Two predicates $\phi$ and $\psi$
are \emph{equivalent} when $\den{\phi}=\den{\psi}$, denoted by
$\phi\equiv\psi$.
If $\varphi\nequiv\bot$ then $\varphi$ is \emph{satisfiable} or
$\SAT{\varphi}$.

\myparagraph{Character classes.}
In all the examples below we let $\D$ stand for the standard 16-bit
character set of Unicode\footnote{Also known as \emph{Plane 0} or
the \emph{Basic Multilingual Plane} of Unicode.}  and use the {.NET}
syntax~\cite{CSharpRegexRef} of regular expression character
classes.  For example, \texttt{[A-Z]} stands for all the Latin capital
letters, \texttt{[0-9]} for all the Latin numerals,
\texttt{\bslash{d}} for all the decimal digits,
\texttt{\bslash{w}} for all the word-letters and
dot (\texttt{.}) for all characters besides the \emph{newline
character} \bslash{n}. (It is a standard convention that,
by default, dot does not match \bslash{n}.). We also use the
$\top$ predicate to denote the set of all characters, which can be
thought of as a canonical representation of
$\texttt{[\bslash{s}\bslash{S}]}$ (or ${\texttt{[\bslash{w}\bslash{W}]}}$).

When we need to distinguish the concrete representation of character
classes from the corresponding abstract representation of predicates
in $\A$ we map each character class $C$ to the corresponding
predicate $\CCpred{$C$}$ in $\Psi$.  For example
$\CCpred{[\caret0-9]}\equiv\lnot\CCpred{[0-9]}$,
$\CCpred{{[\bslash{w}-[\bslash{d}]]}}\equiv\CCpred{\bslash{w}}\land\lnot\CCpred{\bslash{d}}$,
and $\CCpred{\bslash{W}} \equiv \lnot\CCpred{\bslash{w}}$.  Observe
also that
$\den{\CCpred{[0-9]}}\subsetneq\den{\CCpred{\bslash{d}}}\subsetneq\den{\CCpred{\bslash{w}}}$
and $\den{\CCpred{\bslash{n}}}=\{\bslash{n}\}$ and
$\bslash{n}\notin\den{\CCpred{\bslash{w}}}$ because $\bslash{n}$ is
not a word-letter.
All predicates in $\Psi$ are represented in a \emph{canonical} form in the implementation so that
if $\den{\varphi}=\den{\psi}$ then $\varphi=\psi$ by reusing the
predicate algebra for Unicode Plane 0 that is built-in into the {.NET} 7 runtime.
In other words, character classes are only used as the concrete notation,
while, e.g., $\CCpred{\texttt{[\bslash{s}\bslash{S}]}} = \CCpred{\texttt{[\bslash{w}\bslash{W}]}}=\top$.

\section{Regexes with Lookarounds and Location Derivatives}
\label{sec:anchored}
Here we formally define regular expressions with \emph{lookarounds}
and \emph{loops} supporting finite and infinite bounds.  Regexes are
defined modulo a character theory $\A=(\D, \Psi, \den{\_}, \bot,\DOT,
\vee, \wedge, \neg)$ that we illustrate with standard (.NET Regex)
character classes in examples while the actual representation of
character classes in $\Psi$ is irrelevant for the purposes here where
$\D$ may even be infinite.  After the definition of regexes, we
formally define a framework of \emph{derivatives} that leads to the
key notion of \emph{derivation relation} between locations that is
used to prove properties in this framework.  We build on and extend
definitions from~\cite{PLDI2023} by incorporating \emph{intersection},
\emph{complement}, \emph{loops}, and \emph{lookarounds} into a single
unified theory of location derivatives and matching.

\subsection{Regexes}
\label{sec:RE-def}
The class $\RE$ of extended regular expressions with lookarounds,
or \emph{regexes}, is here defined by
the following abstract grammar. Let $\psi\in\Psi$, and $0\leq m \leq
n\neq 0$, or $n=\infty$, and let $R$ range over $\RE$.
\[
\begin{array}{rcl}
R  &::= &
\psi {\quad\mid\quad}
\eps {\quad\mid\quad}
R \alt R' {\quad\mid\quad}
R \rand R' {\quad\mid\quad}
R\cdot R' {\quad\mid\quad}
\RELoop{R}{m}{n} \quad{\mid}\quad
\rnot R  \quad{\mid}
\\[.5em]
&&
\la{R} {\quad\mid\quad} \lb{R} {\quad\mid\quad} \laneg{R} {\quad\mid\quad} \lbneg{R}
\end{array}
\]
where the operators in the first row appear in order of precedence,
with \emph{union} ($\alt$) having lowest and \emph{complement} ($\rnot$) having highest precedence.
The remaining operators in the first row are the \emph{empty-word regex} ($\eps$),
\emph{intersection} ($\rand$), \emph{concatenation} $(\cdot)$, 
and the \emph{loop} expression $\RELoop{R}{m}{n}$ where $R$ is the
\emph{body}, $m$ the \emph{lower bound}, and $n$ the \emph{upper bound} of the loop.
If $n=\infty$ then the loop is \emph{infinite} else \emph{finite}.
We let $\RELoop{R}{0}{0}\eqdef\eps$
for convenience in recursive definitions.
We use the common abbreviations $R\st$ for $\RELoop{R}{0}{\infty}$,
$R\plus$ for $\RELoop{R}{1}{\infty}$.

The regex denoting \emph{nothing} is just the predicate $\bot$.
Concatenation operator $\cdot$ is often implicit by using
juxtaposition.
The expressions in the second row are called \emph{lookarounds}:
$\la{R}$ is \emph{lookahead}, $\lb{R}$ is \emph{lookback}, $\laneg{R}$
is \emph{negative lookahead}, and $\lbneg{R}$ is \emph{negative
lookback}.

The \emph{reverse} $\REV{R}$ of $R\in\RE$ is defined as follows:
\[
\begin{array}{rcl@{\quad}rcl}
\REV{\psi}&\eqdef&\psi & \REV{\la{R}} &\eqdef& \lb{\REV{R}}\\
\REV{\eps}&\eqdef&\eps & \REV{\lb{R}} &\eqdef& \la{\REV{R}}\\
\REV{(R \alt S)}&\eqdef&\REV{R} \alt \REV{S} & \REV{\laneg{R}} &\eqdef& \lbneg{\REV{R}} \\
\REV{(R \rand S)}&\eqdef&\REV{R} \rand \REV{S} & \REV{\lbneg{R}} &\eqdef& \laneg{\REV{R}}\\
\REV{(R \cdot S)}&\eqdef&\REV{S} \cdot \REV{R} \\
\REV{\RELoop{R}{m}{n}}&\eqdef&\RELoop{\REV{R}}{m}{n} \\
\REV{(\rnot R)}&\eqdef&\rnot(\REV{R})
\end{array}
\]
Note that the reverse $\REV{R}$ of a regex $R$ has exactly the same
size as $R$. The size of a regex, $|R|$, is defined recursively as
the number of subexpressions, where each predicate $\psi\in\Psi$ is considered to have size
one, i.e., the actual representation size of predicates is irrelevant
in this context. It also follows by induction from the definition that
$\REV{(\REV{R})}=R$.

\subsection{Derivatives}
\label{sec:derivatives}
In the following definition let $x=\loc{s}{i}$ be a valid nonfinal location.
For example, if $s=\str{ab}$ then the valid nonfinal locations in $s$
are $\loc{s}{0}$ and $\loc{s}{1}$.
In the following let $\ell$ range over lookarounds.
and let $\decr{\infty}\eqdef\infty$, $\decr{0}\eqdef 0$,
and $\decr{k} \eqdef k-1$ for $k> 0$.
\[
\begin{array}{@{}rcl@{}}
\DER{x}{\eps}    &\eqdef& \emp \\
\DER{x}{\ell}    &\eqdef& \emp \\
\DER{\loc{s}{i}}{\psi} &\eqdef& \ITEBrace{s_i\in\den{\psi}}{\eps}{\bot} \\
\DER{x}{R \alt S} &\eqdef& \DER{x}{R}\alt\DER{x}{S} \\
\DER{x}{R \rand S} &\eqdef& \DER{x}{R}\rand\DER{x}{S} \\
\DER{x}{\rnot R} &\eqdef& \rnot\DER{x}{R} \\
\DER{x}{R{\cdot}S} &\eqdef& \ITEBrace{\IsNullable[x]{R}}{\DER{x}{R}{\cdot}S{\alt}\DER{x}{S}}{\DER{x}{R}{\cdot}S}
\\
\DER{x}{\RELoop{R}{m}{n}}
& \eqdef &
\ITEBrace{m{=}0 \bor
\AlwaysNullable{R}{=}\TT \bor
\IsNullable[x]{R}{=}\FF}{
\DER{x}{R}{\cdot}\RELoop{R}{\decr{m}}{\decr{n}}
}{
\DER{x}{R{\cdot}\RELoop{R}{\decr{m}}{\decr{n}}}
}
\end{array}
\]
where $\AlwaysNullable{R}\eqdef\forall x(\IsNullable[x]{R}=\TT)$ intuitively
means that $R$ is nullable and nullability does not depend on lookarounds.
Next we define nullability of locations and what it means to find a
match \emph{end} location from a valid \emph{start} location $x$ in a
string $s$ by a regex $R\in\RE$.  $\FME{x}{R}$ returns the
\emph{latest match end location} from a valid $x$ or $\NoMatch$ if
none exists.  Note that $\max(x,\NoMatch)=\max(\NoMatch,x)=x$.
\[
\begin{array}{rcl}
\NullF{x}{R}&\eqdef&\ITEBrace{\IsNullable[x]{R}}{x}{\NoMatch} \\
\FME{x}{R} &\eqdef& \ITEBrace{\IsFinal{x}}{\NullF{x}{R}}{\max(\NullF{x}{R},\FME{x{+}1}{\DER{x}{R}})}\\
\IsMatch{x}{R}&\eqdef&\FME{x}{R}\neq\NoMatch \\
\IsNullable[x]{\eps} &\eqdef&\TT\\
\IsNullable[x]{\psi} &\eqdef&\FF\\
\IsNullable[x]{R\alt R'} &\eqdef& \IsNullable[x]{R} \bor  \IsNullable[x]{R'}\\
\IsNullable[x]{R\rand R'} &\eqdef& \IsNullable[x]{R} \band  \IsNullable[x]{R'}\\
\IsNullable[x]{R\cdot R'} &\eqdef& \IsNullable[x]{R} \band  \IsNullable[x]{R'}\\
\IsNullable[x]{\RELoop{R}{m}{n}} &\eqdef& m=0 \bor \IsNullable[x]{R} \\
\IsNullable[x]{\rnot R} &\eqdef& \bnot\, \IsNullable[x]{R} \\
\IsNullable[x]{\la{R}} &\eqdef& \IsMatch{x}{R} \\
\IsNullable[x]{\lb{R}} &\eqdef& \IsMatch{\REV{x}}{\REV{R}} \\
\IsNullable[x]{\laneg{R}} &\eqdef& \bnot\,\IsMatch{x}{R} \\
\IsNullable[x]{\lbneg{R}} &\eqdef& \bnot\,\IsMatch{\REV{x}}{\REV{R}}
\end{array}
\]
Observe  that  the  definition  of  $\FME{x}{R}$  with  $x=\loc{s}{i}$
computes  the transition  from the  source  state (regex)  $R$ to  the
target state $S=\DER{x}{R}$ for the character $s_i$ and then continues
matching  from  location  $x+1$  and state  $S$.   The  definition  of
derivatives  and   $\IsMatchName$  are  mutually   recursive,  through
nullability  test  of  lookarounds.  The  definition of $\IsNullable[x]{\ell}$
of lookarounds $\ell$ is  well-defined
because the regex $R$ in $\ell$ is
a strictly smaller  expression  than $\ell$ and
$\REV{R}$ has the same size as $R$.

The definition of $\IsMatch{x}{R}$ terminates as soon as $\IsNullable[x]{R}=\TT$.
This is an obvious but important optimization that is absent from the
abstract definition above, along many other common cases, to avoid redundant
computation.

\subsection{Derivation Relation}

The key concept is the \emph{derivation relation} $x\DERS{R}y$ between locations,
that is instrumental in reasoning about properties of derivatives.
Given a valid location $x$, 
the definition of \emph{all matches of $R$ from $x$}, $\FindAllMatches{x}{R}$,
is as follows.
\[
\begin{array}{r@{\;}c@{\;}l@{\quad}r@{\;}c@{\;}l}
\FindAllMatches{x}{R} &\eqdef&
\ITEBrace{\IsFinal{x}}{\NullE{x}{R}}{\NullE{x}{R}\cup\FindAllMatches{x{+}1}{\DER{x}{R}}} &
\textrm{where}\quad\NullE{x}{R}&\eqdef& \ITEBrace{\IsNullable[x]{R}}{\{x\}}{\emptyset} 
\\
x\DERS{R}y & \eqdef&  y \in \FindAllMatches{x}{R} \\
\end{array}
\]
Two regexes $R$ and $S$ are \emph{equivalent}, denoted by $R\equiv S$,
if, for all locations $x$ and $y$, $x\DERS{R}y \biff x\DERS{S}y$.
The following is the main derivation theorem that
extends~\cite[Theorem~3.3]{PLDI2023}.
When $x$ is a final location we let
$\FindAllMatches{x{+}1}{R}\eqdef\emptyset$ for convenience
in formal arguments. Note also that, for all $x$,
$\FindAllMatches{x}{\emp}=\emptyset$.

\begin{thm}[Derivation] For all regexes and valid locations, let $\ell$ be a lookaround and $\psi\in\Psi$:
\label{thm:DERS}
\begin{enumerate}
\item
\label{thm:DERS:()}
$x\DERS{\eps}y \biff x=y$;
\item
\label{thm:DERS:a}
$x\DERS{\ell}y \biff \IsNullable[x]{\ell}\band x=y$;
\item
\label{thm:DERS:b}
$\loc{s}{i}\DERS{\psi}y \biff s_i\in\den{\psi} \band y = \loc{s}{i+1}$;
\item
\label{thm:DERS:d}
$x\DERS{R{\alt}S}y \biff (x\DERS{R}y \bor x\DERS{S}y)$;
\item
\label{thm:DERS:i}
$x\DERS{R{\rand}S}y \biff (x\DERS{R}y \band x\DERS{S}y)$;
\item
\label{thm:DERS:n}
$x\DERS{\rnot R}y \biff \bnot\,(x\DERS{R}y)$;
\item
\label{thm:DERS:c}
$x\DERS{R{\cdot}S}y \biff \exists z(x\DERS{R}z\DERS{S}y)$;
\item
\label{thm:DERS:g}
$\forall m>0: \RELoop{R}{m}{n} \equiv R{\cdot}\RELoop{R}{m-1}{n-1} \equiv \RELoop{R}{m-1}{n-1}{\cdot}R$;
\item
\label{thm:DERS:f}
$\RELoop{R}{0}{n} \equiv \RELoop{R}{1}{n}{\alt}\eps$;
\end{enumerate}
\end{thm}
\begin{proof}
  \ref{thm:DERS}(\ref{thm:DERS:()}) follows from $\IsNullable[x]{\eps}=\TT$,
  $\DER{x}{\eps}=\emp$ and $\FindAllMatches{x{+}1}{\emp}=\emptyset$.

\myparagraph{Proof of \ref{thm:DERS}(\ref{thm:DERS:a})}
We have that $x\DERS{\ell}y$ iff $y \in \FindAllMatches{x}{\ell}$ iff
$y \in \NullE{x}{\ell}$ because
$\DER{x}{\ell}=\emp$ and $\FindAllMatches{x{+}1}{\emp}=\emptyset$.
We also have that if $\IsNullable[x]{\ell}$ then $\NullE{x}{\ell}=\{x\}$
else $\NullE{x}{\ell}=\emptyset$.

\myparagraph{Proof of \ref{thm:DERS}(\ref{thm:DERS:b})}
Here $x=\loc{s}{i}$ is nonfinal.
It follows that
$\loc{s}{i}\DERS{\psi}y$ iff
$y\in\FindAllMatches{x{+}1}{\DER{x}{\psi}}$
iff $s_i\in\den{\psi}$ and $y\in\NullE{x{+}1}{\eps}$
iff $s_i\in\den{\psi}$ and $y=x{+}1=\loc{s}{i{+}1}$.

\myparagraph{Proof of \ref{thm:DERS}(\ref{thm:DERS:d})}
By induction over $\second{y}-\second{x}$ with $y$ fixed.
If $x=y$ then
\[
\begin{array}{l}
x\DERS{R \alt S}y
\stackrel{(x=y)}{\biff}
\IsNullable[x]{R\alt S} \biff
\IsNullable[x]{R} \bor \IsNullable[x]{S} 
\stackrel{(x=y)}{\biff}
x\DERS{R}y \bor  x\DERS{S}y 
\end{array}
\]
If $x<y$ then
\[
\begin{array}{rl}
x\DERS{R{\alt}S}y
\biff
x{+}1\DERS{\DER{x}{R{\alt}S}}y
&\biff
x{+}1\DERS{\DER{x}{R}{\alt}\DER{x}{S}}y \\
&\stackrel{\textsc{ih}}{\biff}
x{+}1\DERS{\DER{x}{R}}y \bor x{+}1\DERS{\DER{x}{S}}y 
\biff
x \DERS{R} y \bor x\DERS{S}y
\end{array}
\]

\myparagraph{Proof of \ref{thm:DERS}(\ref{thm:DERS:i})}
By induction over $\second{y}-\second{x}$ with $y$ fixed.
For the base case $x=y$ we get that %
\[
\begin{array}{l}
x\DERS{R \rand S}y
\stackrel{(x=y)}{\biff}
\IsNullable[x]{R\rand S} \biff
\IsNullable[x]{R} \band \IsNullable[x]{S} 
\stackrel{(x=y)}{\biff}
x\DERS{R}y \band  x\DERS{S}y 
\end{array}
\]
For the induction case $x<y$ we get that
\[
\begin{array}{rl}
x\DERS{R \rand S}y 
\stackrel{(x<y)}{\biff}
x{+}1\DERS{\DER{x}{R\rand S}}y
&\biff
x{+}1\DERS{\DER{x}{R}\rand \DER{x}{S}}y \\
&\stackrel{\textsc{(ih)}}{\biff}
x{+}1\DERS{\DER{x}{R}}y \band x{+}1\DERS{\DER{x}{S}}y
\stackrel{(x<y)}{\biff}
x\DERS{R}y \band x\DERS{S}y
\end{array}
\]

\myparagraph{Proof of \ref{thm:DERS}(\ref{thm:DERS:n})}
By induction over $\second{y}-\second{x}$.
For the base case $x=y$ we get that %
\[
\begin{array}{l}
x\DERS{\rnot R}y
\stackrel{(x=y)}{\biff}
\IsNullable[x]{\rnot R} \biff
\bnot(\IsNullable[x]{R})
\stackrel{(x=y)}{\biff}
\bnot(x\DERS{R}y)
\end{array}
\]
For the induction case $x<y$ we get that
\[
\begin{array}{l}
x\DERS{\rnot R}y
\stackrel{(x<y)}{\biff}
x+1\DERS{\DER{x}{\rnot R}}y \biff
x+1\DERS{\rnot\DER{x}{R}}y
\stackrel{\textsc{(ih)}}{\biff}
\bnot(x+1\DERS{\DER{x}{R}}y)
\stackrel{(x<y)}{\biff}
\bnot(x\DERS{R}y)
\end{array}
\]

\myparagraph{Proof of \ref{thm:DERS}(\ref{thm:DERS:c})} by induction over $\second{y}-\second{x}$ with $y$ fixed.
If $x=y$ then
\[
x\!\DERS{R{\cdot}S}\!y \biff
y\in\NullE{x}{R{\cdot}S}\biff (y\in\NullE{x}{R}\band y\in\NullE{x}{S})
\biff (x\!\DERS{R}\!y\!\DERS{S}\!y) \biff \exists z(x\!\DERS{R}\!z\!\DERS{S}\!y)
\]
Observe that $z$ in the last
`$\biff$' must be $x$ because $x\DERS{R}z\DERS{S}x$ implies that $x\leq z \leq x$.
If $x<y$ then, by case analysis over $\IsNullable[x]{R}$, if  $\IsNullable[x]{R}=\FF$ then
\begin{align*}
x\DERS{R{\cdot}S}y \biff
x{+}1\DERS{\DER{x}{R{\cdot}S}}y 
\biff
x{+}1\DERS{\DER{x}{R}{\cdot}S}y 
\stackrel{\textsc{ih}}{\biff}
\exists z(x{+}1\DERS{\DER{x}{R}}z\DERS{S}y)
\biff
\exists z(x\DERS{R}z\DERS{S}y)
\end{align*}
In the last `$\biff$' above $\IsNullable[x]{R}=\FF$ so $x{+}1\leq z$.
If $\IsNullable[x]{R}=\TT$ then
\[
\begin{array}{rcl}
x\DERS{R{\cdot}S}y &\biff&
x{+}1\DERS{\DER{x}{R}{\cdot}S{\alt}\DER{x}{S}}y \\
&\stackrel{\textrm{by \ref{thm:DERS}(\ref{thm:DERS:d})}}{\biff}&
x{+}1\DERS{\DER{x}{R}{\cdot}S}y \bor
x{+}1\DERS{\DER{x}{S}}y \\
&\stackrel{\textsc{ih}}{\biff}&
\exists z(x{+}1\DERS{\DER{x}{R}}z\DERS{S}y) \bor
x{+}1\DERS{\DER{x}{S}}y \\
&\stackrel{(x<y)}{\biff}&
\exists z(x{+}1\DERS{\DER{x}{R}}z\DERS{S}y) \bor
x\DERS{S}y \\
&\biff&
\exists z{>}x(x\DERS{R}z\DERS{S}y) \bor
x\DERS{S}y \\
&\stackrel{\IsNullable[x]{R}=\TT}{\biff}&
\exists z(x\DERS{R}z\DERS{S}y) 
\end{array}
\]

\myparagraph{Proof of \ref{thm:DERS}(\ref{thm:DERS:g}).}
Let $L=\RELoop{R}{m}{n}$ where $m>0$.
We prove the statement by proving (for all $m$ and $n$)
\[
x\DERS{L}y \biff x \DERS{R{\cdot}(L-1)} y
\]
The statement holds trivially when $m=n=1$. Assume $n>1$.

We prove the statement by induction over $\second{y}-\second{x}$.
In each induction step we also get, by using the IH, 
that $R(L-1)\equiv R((L-2)R) \equiv (R(L-2))R \equiv(L-1)R$,
where $(L-2)$ is well-defined because $n>1$.

The case $\AlwaysNullable{R}=\FF$ but $\IsNullable[x]{R}=\TT$ then
$\DER{x}{L}=\DER{x}{R{\cdot}(L-1)}$ by definition.
If $\IsNullable[x]{R}=\FF$ then 
$\DER{x}{L}=\DER{x}{R}{\cdot}(L-1)$ but then also
$\DER{x}{R{\cdot}(L-1)}=\DER{x}{R}{\cdot}(L-1)$. So, in either case
it follows that $L\equiv R{\cdot}(L-1)$ without induction.

The remaining case is $\AlwaysNullable{R}=\TT$.
The base case of $x=y$ follows directly because $x=y$ iff 
$\IsNullable[x]{R}=\TT$ iff
$\IsNullable[x]{L}$ iff $\IsNullable[x]{R{\cdot}(L-1)}$.

For the induction case Let $x<y$.
If $m>1$ then let $L' = (L-1)$ else if $m=1$, using that
$\RELoop{R}{0}{k}\equiv\RELoop{R}{1}{k}$ when
$\AlwaysNullable{R}=\TT$, let $L'=\RELoop{R}{1}{n-1}$ in order to
maintain that the lower bound remains positive. Observe that the
induction is over $\second{y}-\second{x}$ and in the induction steps
there exists $z\geq x+1$ such that $z\DERS{L'}y$ where
$\second{y}-\second{z} < \second{y}-\second{x}$.
\[
\begin{array}{rcl}
x\DERS{R{\cdot}(L-1)}y
&\biff&
x{+}1 \DERS{\DER{x}{R{\cdot}{L'}}}y
\\
&\biff& 
x{+}1 \DERS{\DER{x}{R}{\cdot}{L'}\alt \DER{x}{L'}}y
\\
&\biff&
x{+}1 \DERS{\DER{x}{R}{\cdot}{L'}}y
\bor x{+}1 \DERS{\DER{x}{L'}} y
\\
&\stackrel{\textsc{(ih)}}{\biff}&
x{+}1 \DERS{\DER{x}{R}{\cdot}(L-2){\cdot}R}y
\bor x{+}1 \DERS{\DER{x}{L'}} y
\\
&\biff& 
x{+}1 \DERS{\DER{x}{L'}{\cdot}R} y
\bor x{+}1 \DERS{\DER{x}{L'}} y
\\
&\stackrel{\AlwaysNullable{R}}{\biff}& 
x{+}1 \DERS{\DER{x}{L'}{\cdot}R} y
\\
&\biff&
x{+}1 \DERS{\DER{x}{R}{\cdot}(L-2){\cdot}R}y
\\
&\stackrel{\textsc{(ih)}}{\biff}&
x{+}1 \DERS{\DER{x}{R}{\cdot}(L')}y
\\
&\biff&
x\DERS{L} y
\end{array}
\]

\myparagraph{Proof of \ref{thm:DERS}(\ref{thm:DERS:f})}
for the case of $\AlwaysNullable{R}=\TT$ is similar to the proof of \ref{thm:DERS}(\ref{thm:DERS:g})
and observe that in this case $R\equiv R{\alt}\eps$.
Assume $\AlwaysNullable{R}=\FF$.
Then for $\IsNullable[x]{R}=\TT$ or
$\IsNullable[x]{R}=\FF$ the proof follows directly because in either case
$\DER{x}{\RELoop{R}{0}{n}}=\DER{x}{R{\cdot}\RELoop{R}{0}{n-1}}$:
\[
x\DERS{\RELoop{R}{0}{n}}y \biff
x+1\DERS{\DER{x}{R{\cdot}\RELoop{R}{0}{n-1}}}y \bor x\DERS{\eps}y \biff x\DERS{R{\cdot}\RELoop{R}{0}{n-1}\alt\eps}y
\]
where $R{\cdot}\RELoop{R}{0}{n-1}\equiv\RELoop{R}{1}{n}$ by \ref{thm:DERS}(\ref{thm:DERS:g}).
The theorem follows by the induction principle over long distances.
\end{proof}
It is possible to prove Theorem~\ref{thm:DERS} by induction over $R$ but the
proof becomes much more involved, while using
induction over location distances takes full advantage of the
definition of location derivatives.  We get the following key
characterization of $\ERE$ as a corollary of Theorem~\ref{thm:DERS}.
Let
\[
\begin{array}{rcl}
  \MatchUniverse &\eqdef& \{\pair{\loc{s}{i}}{\loc{s}{j}} \mid s\in\D^*, 0\leq i\leq j \leq |s|\} \\
  \pair{x}{y}\models R &\eqdef& x\DERS{R}y \quad \textrm{for}\;\pair{x}{y}\in\MatchUniverse\\
\LANGRE{R} &\eqdef& \{M\in\MatchUniverse\mid M\models R\}
\end{array}
\]
We say that $M\in\MatchUniverse$ is a \emph{match} of $R$ if $M\models R$,
and $\LANGRE{R}$ is called the \emph{match set} or \emph{match sematics} of $R$.
Observe that $R\equiv S \biff \LANGRE{R}=\LANGRE{S}$.
\begin{cor}
\label{cor:BA}
$\MatchAlgebra = (\MatchUniverse, \ERE, \LANGREName, \bot, \top\st, \alt, \rand, \rnot)$
is an effective Boolean algebra over $\MatchUniverse$.
\end{cor}
\begin{proof}
Recall that \emph{complement} of $X\subseteq \MatchUniverse$ in $\MatchAlgebra$
is $\COMPL{X}\eqdef\MatchUniverse\setminus X$.
It holds that $\LANGRE{\bot} = \emptyset$ and $\LANGRE{\top\st} = \MatchUniverse$.
Theorem~\ref{thm:DERS}(\ref{thm:DERS:d}) implies that
$\LANGRE{R\alt S} = \LANGRE{R}\cup\LANGRE{S}$.
Theorem~\ref{thm:DERS}(\ref{thm:DERS:i}) implies that
$\LANGRE{R\rand S} = \LANGRE{R}\cap\LANGRE{S}$.
Theorem~\ref{thm:DERS}(\ref{thm:DERS:n}) implies that
$\LANGRE{\rnot R} = \COMPL{\LANGRE{R}}$.
\end{proof}
Corollary~\ref{cor:BA} is a powerful tool that allows us to
apply laws of distributivity and de Morgan's laws in order to
rewrite regexes into various normal forms and to simplify regexes
based on other general laws of Boolean algebras as well as specific
laws of $\MatchAlgebra$ in combination with
the laws of Kleene-star and concatenation.
For example, it follows from Theorem~\ref{thm:DERS} that
$\rnot\bot \equiv \top\st$,
$\rnot\eps\equiv\top\plus$, 
$\top\st \rand R \equiv R$, and thus for example that
$\DER{x}{\rnot\eps\rand R}\equiv\DER{x}{\top\plus}\rand\DER{x}{R}
=\top\st\rand\DER{x}{R} \equiv \DER{x}{R}$.

\subsection{Reversal Theorem}

Reversal of regexes plays a key role in the top-level matching algorithm,
where it is used to locate the start location of a match backwards from the
end location (provided that the end location exists). Here our focus is on
reversal itself and we extend~\cite[Theorem~3.8]{PLDI2023}.
We make use of the following basic semantic fact of lookbacks.
\begin{lma}
  \label{lma:lb}
  Let $R\in\RE$ and $x$ be a valid location. Then
  $\IsNullable[x]{\lb{R}}\biff\exists z:z\DERS{R}x$.
\end{lma}
\begin{thm}[Reversal]
\label{thm:REV}
For all $R\in\RE$ and valid locations $x$ and $y$:
\begin{enumerate}
\item
 $\IsNullable[x]{R} \biff  \IsNullable[\REV{x}]{\REV{R}}$
\item
 $x\DERS{R}y\biff\REV{y}\DERS{\REV{R}}\REV{x}$.
\end{enumerate}
\end{thm}
\begin{proof}
The proof uses Theorem~\ref{thm:DERS} and is by induction over $(|R|,\second{y}-\second{x})$
where  $(|S|,d)<(|R|,e)$ if either $|S| < |R|$ or $|S|=|R|$ and $d<e$.
For loops $|\RELoop{X}{m}{n}|<|\RELoop{X}{m'}{n'}|$ when
either $m<m'$ or else $m=m'$ and $n<n'$.
Induction uses $\second{y}-\second{x}$ only when $R=X\st$.
(Recall also that  $\RELoop{R}{0}{0}\eqdef\eps$.)
Let $x=\loc{s}{i}$ be a valid location.

The proof of (1) follows from (2) in all induction cases except when $R$ is a lookaround
where (1) is proved by using the IH of (2).

\myparagraph{Base case $R=\eps$} Follows from $\REV{\eps}=\eps$.

\myparagraph{Base case $R=\psi$} Here $x$ is \emph{nonfinal} so $x+1$ is valid. We get that
$x\DERS{\psi}y\biff\REV{y}\DERS{\REV{\psi}}\REV{x}$ by~\textsc{Thm~\ref{thm:DERS}(\ref{thm:DERS:b})}
because $x+1=y$ iff $\REV{y}+1 = \REV{x}$ and $\REV{\psi}=\psi$,
where $\REV{y}=\REV{\loc{s}{i+1}}=\loc{\REV{s}}{|s|-i-1}$ and $\REV{s}_{|s|-i-1} = s_i$,
so $s_i\in\den{\psi}$ iff $\REV{s}_{|s|-i-1}\in\den{\REV{\psi}}$.
Note also that $\IsNullable[x]{\psi}=\IsNullable[\REV{x}]{\REV{\psi}}=\IsNullable[\REV{x}]{\psi}=\FF$.

\myparagraph{Induction case $R$ is a lookback} Let $R=\lb{S}$.
\[
\IsNullable[x]{R}  \stackrel{\textsc{(Lma~\ref{lma:lb})}}{\biff}
\exists z:z\DERS{S}x \stackrel{\textsc{(ih)}}{\biff}
\exists z:\REV{x}\DERS{\REV{S}}\REV{z}
\biff
\IsMatch{\REV{x}}{\REV{S}} \biff
\IsNullable[\REV{x}]{\la{\REV{S}}} \biff
\IsNullable[\REV{x}]{\REV{R}}
\]
It follows also that $x\DERS{R}y \biff \REV{y}\DERS{\REV{R}}\REV{x}$
because $x\DERS{R}y\biff(\IsNullable[x]{R}\band x=y)$.
The proof of the negative lookback is similar because
$\IsNullable[x]{\lbneg{S}}\biff \bnot\,\IsNullable[x]{\lb{S}}$.

\myparagraph{Induction case $R$ is a lookahead} Let $R=\la{S}$.
\[
\IsNullable[x]{R} \biff
\IsMatch{x}{S} \biff
\exists z:x\DERS{S}z \stackrel{\textsc{(ih)}}{\biff}
\exists z:\REV{z}\DERS{\REV{S}}\REV{x}
\stackrel{\textsc{(Lma~\ref{lma:lb})}}{\biff}
\IsNullable[\REV{x}]{\lb{\REV{S}}} \biff
\IsNullable[\REV{x}]{\REV{R}}
\]
It follows also that $x\DERS{R}y \biff \REV{y}\DERS{\REV{R}}\REV{x}$.
The proof of the negative lookahead is similar because
$\IsNullable[x]{\laneg{S}}\biff \bnot\,\IsNullable[x]{\la{S}}$.

\myparagraph{Induction case $R=X\!{\alt}\!Y$} Then
\[
x \DERS{R} y
\stackrel{\textsc{Thm~\ref{thm:DERS}(\ref{thm:DERS:d})}}{\biff} x \DERS{X} y \bor x \DERS{Y} y
\stackrel{\textsc{(ih)}}\biff
\REV{y} \DERS{\REV{X}} \REV{x} \bor \REV{y} \DERS{\REV{Y}} \REV{x}
\stackrel{\textsc{Thm~\ref{thm:DERS}(\ref{thm:DERS:d})}}{\biff}
\REV{y} \DERS{\REV{X}{\alt}\REV{Y}} \REV{x}
\biff \REV{y} \DERS{\REV{R}} \REV{x}
\]

\myparagraph{Induction case $R= X{\rand} Y$} Then
\[
x\DERS{R}y \stackrel{\textsc{Thm~\ref{thm:DERS}(\ref{thm:DERS:i})}}{\biff}
x\DERS{X}y\band x\DERS{Y}y \eqIff
\REV{y}\DERS{\REV{X}}\REV{x}\band \REV{y}\DERS{\REV{Y}}\REV{x}
\stackrel{\textsc{Thm~\ref{thm:DERS}(\ref{thm:DERS:i})}}{\biff}
\REV{y}\DERS{\REV{X}\rand\REV{Y}}\REV{x}
\biff
\REV{y}\DERS{\REV{R}}\REV{x}
\]

\myparagraph{Induction case $R=\rnot X$} Then
\[
x\DERS{R}y \stackrel{\textsc{Thm~\ref{thm:DERS}(\ref{thm:DERS:n})}}{\biff}
\bnot(x\DERS{X}y) \eqIff
\bnot(\REV{y}\DERS{\REV{X}}\REV{x})
\stackrel{\textsc{Thm~\ref{thm:DERS}(\ref{thm:DERS:n})}}{\biff}
\REV{y}\DERS{\rnot(\REV{X})}\REV{x}
\biff
\REV{y}\DERS{\REV{R}}\REV{x}
\]

\myparagraph{Induction case $R=X{\cdot}Y$} Then
\[
x \DERS{R} y
\stackrel{\textsc{Thm~\ref{thm:DERS}(\ref{thm:DERS:c})}}{\biff} \exists z(x \DERS{X} z \DERS{Y} y)
\stackrel{\textsc{(ih)}}\biff
\exists z(\REV{y} \DERS{\REV{Y}} \REV{z} \DERS{\REV{X}} \REV{x})
\stackrel{\textsc{Thm~\ref{thm:DERS}(\ref{thm:DERS:c})}}{\biff}
\REV{y} \DERS{\REV{Y}{\cdot}\REV{X}} \REV{x}
\biff \REV{y} \DERS{\REV{R}} \REV{x}
\]

\myparagraph{Induction case $R=\RELoop{X}{m}{n}$ ($m>0$)}
Then%
\[
x\DERS{R}y
\stackrel{\textsc{Thm~\ref{thm:DERS}(\ref{thm:DERS:c},\ref{thm:DERS:g})}}{\biff}
\exists z(x\DERS{X}z\DERS{R-1}y)
\stackrel{\textsc{(ih)}}{\biff}
\exists z(\REV{y}\DERS{\REV{R}-1}\REV{z}\DERS{\REV{X}}\REV{x})
\stackrel{\textsc{Thm~\ref{thm:DERS}(\ref{thm:DERS:c},\ref{thm:DERS:g})}}{\biff}
\REV{y}\DERS{\REV{R}}\REV{x}
\]

\myparagraph{Induction case $R=\RELoop{X}{0}{n}$}
Then, by using Thm~\ref{thm:DERS}(\ref{thm:DERS:g},\ref{thm:DERS:f},\ref{thm:DERS:c},\ref{thm:DERS:a}),
\[
\begin{array}{rcl}
x\DERS{R}y &\biff& 
\exists z(x\DERS{X}z\DERS{R-1}y) \bor x\DERS{\eps}y \\
&\stackrel{\textsc{(ih)}}{\biff}&
\exists z(\REV{y}\DERS{\REV{R}-1}\REV{z}\DERS{\REV{X}}\REV{x})
\bor \REV{y}\DERS{\eps}\REV{x}
\\
&\biff&
\REV{y}\DERS{(\REV{R}-1){\cdot}\REV{X}\alt\eps}\REV{x}
\\
&\biff&
\REV{y}\DERS{\RELoop{\REV{X}}{1}{n}\alt\eps}\REV{x}
\\
&\biff&
\REV{y}\DERS{\REV{R}}\REV{x}
\end{array}
\]
where $\second{y}-\second{z} < \second{y}-\second{x}$ when
$\IsNullable[x]{X}=\FF$ and $n=\infty$, so the IH applies.
If $R=X\st$ and $\IsNullable[x]{X}=\TT$ then
the case $x=y$ is immediate. Otherwise $x<y$ and we may assume that
any ``stuttering'' steps $x\DERS{X}x$ are omitted from the derivation
because they do not contribute to completing the derivation
that is finite and bounded by $\second{y}-\second{x}$. Then
the IH applies because then $z\geq x+1$.
\end{proof}

\subsection{Anchors as Lookarounds}
\label{sec:anchors-as-lookarounds}

All anchors can be defined by lookarounds as was conjectured
in~\cite[Section~3.7]{PLDI2023}.  In our implementation all the
standard anchors are defined by using lookarounds.  Many other custom
anchors can be defined with lookarounds, e.g., to support different
line separators besides $\bslash{n}$ such as $\bslash{r}\bslash{n}$
for different file formats, or to support anchors for detecting
enclosing parentheses.

\begin{table}
  \caption{Standard anchors in regexes and
    their definitions by lookarounds relative to a valid location $\loc{s}{i}$.}
  \label{tab:anchors}
  \begin{tabular}{@{}llll@{}}
    \multicolumn{1}{@{}c@{}}{\ANCH} & Name & Definition & Effective meaning\\
    \hline
     \texttt{{\bslash{A}}} & start & $\lbneg{\top}$ & \emph{initial} location of input ($i=0$)\\
    \texttt{{\bslash{z}}} & end   &  $\laneg{\top}$ & \emph{final} location of input ($i=|s|$)\\
    \caret & start-of-line &  \texttt{({\textbackslash}A|(?<={\textbackslash}n))} & $i=0$ or $s_{i-1}=\bslash{n}$\\
     \texttt{\$} & end-of-line & \texttt{(\bslash{z}|(?={\textbackslash}n))} &  $i=|s|$ or $s_i=\bslash{n}$ \\
    \texttt{\bslash{Z}} &
     end-or-last-\texttt{\dollar} &
    \texttt{(\bslash{z}|\la{\bslash{n}${\cdot}$\bslash{z}})} &
    $i=|s|$ or $i=|s|{-}1$ and $s_i=\bslash{n}$  \\
    \texttt{{\textbackslash}b} &
    word-border & 
    $\lb{\bslash{w}}{\cdot}\laneg{\bslash{w}}{\alt}\lbneg{\bslash{w}}{\cdot}\la{\bslash{w}}$
    &
    $s_{i-1}\in\den{\CCpred{\bslash{w}}} \IFF s_i\notin\den{\CCpred{\bslash{w}}}$
    \\
    \texttt{{\bslash{B}}} &
    non-word-border & 
     $\lb{\bslash{w}}{\cdot}\la{\bslash{w}}{\alt}\lbneg{\bslash{w}}{\cdot}\laneg{\bslash{w}}$
    &
    $s_{i-1}\in\den{\CCpred{\bslash{w}}} \IFF s_i\in\den{\CCpred{\bslash{w}}}$
    \\    \hline
  \end{tabular}
\end{table}
Table \ref{tab:anchors} gives a summary of all anchors typically
supported by regular expression engines~\cite{friedl2006mastering}
represented in terms of lookarounds. Recall that
$s_{-1}=s_{|s|}=\epsilon$ and thus, for all $\psi\in\Psi$,
$s_{-1}=s_{|s|}\notin\den{\psi}$. Therefore, e.g., the definition of
$\bslash{B}$ as
$\lb{\bslash{w}}\la{\bslash{w}}{\alt}\lb{\bslash{W}}\la{\bslash{W}}$
would be \emph{incorrect} where $\den{\CCpred{\bslash{W}}}=
\den{\lnot\CCpred{\bslash{w}}}$, because for example $\bslash{B}$
is nullable in \emph{all locations} of $\texttt{"\#"}$ instead of in \emph{no locations}
if one would use $\lb{\bslash{w}}\la{\bslash{w}}{\alt}\lb{\bslash{W}}\la{\bslash{W}}$.

In other words,
$\laneg{\bslash{w}}\nequiv\la{\bslash{W}}$. In general, negative
lookarounds of predicates \emph{cannot be converted} into positive
lookarounds by delegating negation to the underlying character algebra
$\A$ by using $\lnot$ of $\A$.

It is also important to note that \emph{complement} $(\rnot)$ of
positive lookahead does not correspond to a negative lookahead (and
vice versa).  Below we show that $\rnot\bslash{b}\nequiv\bslash{B}$,
i.e. they are not complements of each other in match semantics. Thus,
negations of combinations of lookarounds must be used with caution as
their semantics needs careful analysis.  In fact,
$\rnot\bslash{b}\equiv\bslash{B}{\alt}\top\plus$, as shown below.

While it follows from definitions, for all valid locations $x$, that
$\IsNullable[x]{\laneg{R}} \biff \bnot\,\IsNullable[x]{\la{R}}$
their \emph{derivatives} are fundamentally different:
they are \emph{complements} of each other. In particular,
\[
\begin{array}{rcl}
\DER{x}{\laneg{R}} &=&\bot \\
\DER{x}{\rnot{\la{R}}} &=& \rnot\DER{x}{\la{R}} = \rnot\bot \equiv \top\st
\end{array}
\]
The same applies to lookbacks. The derivative of a lookaround is always $\bot$ because
lookarounds are context conditions for locations with
no forward progress on their own, and essentially operate as blocking conditions inside concatenations.
It follows from the definition of derivatives of concatenations that, if $\ell$ is a lookaround then
\[
\DER{x}{\ell{\cdot}R} \equiv
\ITEBrace{\IsNullable[x]{\ell}}{\DER{x}{R}}{\bot}
\]
which is how derivatives are implemented for this case, i.e., the
simplification rewrite rules are directly embedded into the abstract
definition itself, rather than being applied afterwards in a separate
step. The same applies to all similar situations involving simplification
rules for example for $\eps$, $\bot$, and $\top\st$, in intersections,
unions, and concatenations. Many of those rules are \emph{mandatory} in order
to maintain finiteness of the state space of regexes that arise during
matching, and all rewrite rules are applied as early as possible.

In order to understand the interaction of anchors with newly added
extensions better, let us take the word-border anchor
\texttt{{\textbackslash}b} and derive what the complement of it
$(\rnot\texttt{{\textbackslash}b})$ would look like.  Here we make use
of de Morgan's laws via Corollary~\ref{cor:BA} and the property that
concatenation of two lookarounds has the same semantics as the
intersection (conjunction) of those lookarounds.
Moreover, to simplify, we use laws of distributivity and the
properties that
\[
\begin{array}{rcl}
\rnot\laneg{R}\rand\rnot\la{R}&\equiv&
\top\plus
\\
\rnot\lbneg{R}\rand\rnot\lb{R}&\equiv&
\top\plus
\end{array}
\]
where $\DER{x}{\rnot\ell\rand\rnot\ell'}\equiv\top\st$ and
the combined nullability conditions become trivially false, e.g., 
\[
\IsNullable[x]{\rnot\laneg{R}\rand\rnot\la{R}}
\biff
\IsNullable[x]{\la{R}}\band \bnot\,\IsNullable[x]{\la{R}}
\]
Then
\[
\begin{array}{@{}r@{\;}c@{\;}l@{}}
  \rnot\bslash{b} &\equiv&
  \rnot(\lb{\bslash{w}}{\rand}\laneg{\bslash{w}}{\alt}\lbneg{\bslash{w}}{\rand}\la{\bslash{w}})
  \\
  &\equiv&
  (\rnot\lb{\bslash{w}}{\alt}\rnot\laneg{\bslash{w}}){\rand}
  (\rnot\lbneg{\bslash{w}}{\alt}\rnot\la{\bslash{w}})
  \\
  &\equiv&
  \rnot\lb{\bslash{w}}{\rand}\rnot\lbneg{\bslash{w}} \alt
  \rnot\lb{\bslash{w}}{\rand}\rnot\la{\bslash{w}} \alt
  \rnot\laneg{\bslash{w}}{\rand}\rnot\lbneg{\bslash{w}} \alt
  \rnot\laneg{\bslash{w}}{\rand}\rnot\la{\bslash{w}}
  \\
  &\equiv&
  \rnot\lb{\bslash{w}}{\rand}\rnot\la{\bslash{w}} \alt
  \rnot\laneg{\bslash{w}}{\rand}\rnot\lbneg{\bslash{w}} \alt
  \top\plus
  \\
  &\equiv&
  \bslash{B}\alt\top\plus
\end{array}
\]
where the last equivalence holds because, for any valid location $x$,
\[
\IsNullable[x]{\rnot\lb{\bslash{w}}{\rand}\rnot\la{\bslash{w}} \alt
  \rnot\laneg{\bslash{w}}{\rand}\rnot\lbneg{\bslash{w}}}
\IFF
\IsNullable[x]{\lbneg{\bslash{w}}{\rand}\laneg{\bslash{w}} \alt
  \la{\bslash{w}}{\rand}\lb{\bslash{w}}}
\]
that is equivalent to $\IsNullable[x]{\bslash{B}}$.
One can also show that $\rnot\bslash{B}\equiv\bslash{b}\alt\top\plus$.
Observe therefore that
\[
\rnot(\bslash{B}\alt\top\plus) \equiv \rnot\bslash{B}\rand \rnot(\top\plus)
\equiv (\bslash{b}\alt\top\plus) \rand\eps
\equiv \bslash{b}\rand\eps \equiv \bslash{b}
\]
A natural question that arises, that we are investigating,
is if there exists a general way to encode any negative lookaround
$\laneg{R}$ (or $\lbneg{R}$) by an equivalent positive lookaround
$\la{R'}$ (or $\lb{R'}$), by converting $R\in\RE$ into a regex
$R'\in\RE$ by using $\rnot$ in some manner.
A possible algorithmic advantage is that this
transformation may enable further transformations.
For example, $\laneg{R}{\cdot}\la{S}$ could then be transformed
into $\la{R'{\cdot}\top\st\rand S{\cdot}\top\st}$.

The anchor denoting the end of
previous match \texttt{\bslash{G}}~\cite{friedl2006mastering} is a non-traditional
anchor as it uses metadata instead of regular expressions. While it
can be supported by adding implementation details, we have decided to
omit it as we have not observed it being pervasively used in real life
applications.
The anchor $\bslash{a}$ in~\cite{PLDI2023} that is needed internally for reversal,
can be defined here as $\REV{(\bslash{Z})}$;
$\bslash{a}$ is not supported in the concrete
syntax of {.NET} regular expressions.

\subsection{Top-Level Match Algorithm}

The top-level match algorithm $\FindMatch{s}{R}$ takes a string
$s\in\D^*$ and a regex $R\in\RE$, and either returns $\NoMatch$ if
$\LANGRE{R}=\emptyset$, or else returns a match
$\pair{\loc{s}{i}}{\loc{s}{j}}\in \LANGRE{R}$ such that $i$ is minimal
(leftmost) and then $j$ is maximal for the given $i$ (in particular
all loops are eager).  Semantically the implementation corresponds to
the following algorithm, which therefore provides the leftmost and longest
match, also known as POSIX semantics,
\[
\begin{array}{rcl}
\FindMatch{s}{R}&\eqdef&
\blet\, x = \FindMatchEnd{\loc{\REV{s}}{0}}{\top\st{\cdot}\REV{R}}\,\bin
\ITEBrace{x=\NoMatch}{\NoMatch}{\pair{\REV{x}}{\FindMatchEnd{\REV{x}}{R}}}
\end{array}
\]
In the actual implementation the first phase initially locates the end
location $\loc{s}{j}$ by simulating a PCRE ``lazy'' loop $\top\st\lazy$
concatenated with $R$, and subsequently locates the start location
backwards from $\loc{s}{j}$ by using $\REV{R}$. It is typically
more economical to start the search from the start of the string
rather than backwards from the end of the string, although both
variants are possible.
\begin{example}
  Consider $s=\texttt{"\#\#\#abacarabacaraba\#\#"}$ and $R=\texttt{abacaraba}$.
  So $\REV{R}=\texttt{abaracaba}$ and $\REV{s}=\texttt{"\#\#abaracabaracaba\#\#\#"}$.
  Then
  $\FindMatchEnd{\loc{\REV{s}}{0}}{\top\st{\cdot}\REV{R}}=\loc{\REV{s}}{17}=\REV{\loc{s}{|s|-17}}=\REV{\loc{s}{3}}$
  and $\FindMatchEnd{\loc{s}{3}}{R}=\loc{s}{12}$.
  So $\FindMatch{s}{R}=\pair{\loc{s}{3}}{\loc{s}{12}}$. Observe that
  $\FindMatchEnd{s}{\top\st{\cdot}R}=\loc{s}{18}$ would be the end location of the \emph{second} match.
  Therefore, when starting with the end location search, the initial $\st$-loop must be
  treated or simulated as a lazy loop.

  If we instead let $R=\texttt{abacaraba\bslash{b}}$ so that the match must end a word then
  in $\REV{R}=\texttt{\bslash{b}abaracaba}$ the match must start a word. In that case
  $\FindMatchEnd{\loc{\REV{s}}{0}}{\top\st{\cdot}\REV{R}}=\loc{\REV{s}}{11}=\REV{\loc{s}{9}}$
  because \texttt{\#} is not a word-letter, and then
  $\FindMatch{s}{R}=\pair{\loc{s}{9}}{\loc{s}{18}}$.
\end{example}
The \emph{POSIX match for $R$ in $s$}, if a match exists,
is a pair $\pair{\loc{s}{i}}{\loc{s}{j}}$ where
$\loc{s}{i}$ is the minimal location in $s$ such that $\exists y:\loc{s}{i}\DERS{R}y$
and $\loc{s}{j}$ is the maximal location in $s$ such that $\loc{s}{i}\DERS{R}\loc{s}{j}$.
The following is the correctness theorem of $\FindMatchName$.
\begin{thm}[Match]
  \label{thm:Match}
  For all $s\in\D^*$ and $R\in\RE$ the following statements hold:
  \begin{enumerate}
  \item $\FindMatch{s}{R}=\NoMatch\biff \nexists i,j:\loc{s}{i}\DERS{R}\loc{s}{j}$;
  \item if $\FindMatch{s}{R}\neq\NoMatch$ then $\FindMatch{s}{R}$
    is the POSIX match for $R$ in $s$.
  \end{enumerate}
\end{thm}
\begin{proof}
  By using Theorems~\ref{thm:DERS} and~\ref{thm:REV}, and that
  for all $x$ and $S\in\RE$,
  $\FindMatchEnd{x}{S}=\max\FindAllMatches{x}{S}$.
  The proof also uses the fact that
  if $\REV{\loc{s}{i}}=\max\FindAllMatches{\loc{\REV{s}}{0}}{\top\st{\cdot}\REV{R}}$
  then $\loc{s}{i}$ is the \emph{minimal} location in $s$ such that
  $\loc{s}{i}\DERS{R{\cdot}\top\st}\loc{s}{|s|}$.
\end{proof}

\section{Implementation}
\label{sec:implementation}

Here we give a high level overview of how the engine was implemented, what parts are 
shared with the rest of .NET and the nonbacktracking engine, and some key differences 
and implementation details.

The core parser was taken directly from the .NET runtime, but was
modified to read the symbols \texttt{\&} and \texttt{$\sim$} as
intersection and negation respectively. Fortunately the escaped
variants `\texttt{\bslash{}\&}' and `\texttt{\bslash{}$\sim$}' were
not assigned to any regex construct, and all existing regex patterns
can be used by escaping these two characters.

Another important component we directly took from the .NET runtime is
the alphabet compression of the nonbacktracking engine that builds the
minimal number of predicates on the underlying character set by
applying appropriate Boolean combinations.  For example, all input
characters in the the regex \texttt{ab} can be represented with just 3
predicates: $a$, $b$ and $[{\caret}ab]$. It
is an essential preprocessing step for such a large alphabet (16-bit
Unicode Basic Multilingual Plane in our case), which allows us to
represent the input predicates as bits, and then use bitwise
operations to check if a character is in the set of characters
represented by a predicate.

The engine itself is implemented in mostly high-level F\#, without
caching any derivatives, which does not give it the optimal
performance characteristics as of at the time of writing. The aim is
rather to provide a working prototype that, despite having much higher
number of memory allocations than, e.g., the nonbacktracking or
backtracking alternatives of .NET 7, is still capable of solving
problems that neither the backtracking nor nonbacktracking variant
can currently solve.

To make the modifications possible, the
\texttt{System.Text.RegularExpressions} library was copied to another
namespace and the visibility of the library was changed from internal
to public. The only other meaningful modification to the library is
the extension of the parser to parse negation and intersection symbols
appropriately, as described above.

The concatenation and $\epsilon$-nodes were implemented as just a
single-linked list of other regex nodes, as it allows to traverse to
next one using just a tail pointer, without needing any other
steps. There is never any need to traverse the concatenation backwards,
thus allowing to take advantage of what the concatenation
fundamentally is -- a single linked list. And $\eps$ is
just an empty concatenation, i.e. empty list.

\subsection{Rewrite Rules and Subsumption}
\label{sec:rewrites}

Our system implements a number of regex rewrite rules, which are
essential for the efficiency of the implementation.
Figure~\ref{fig:rw} illustrates the basic rewrite rules that are
always applied when regular expressions are constructed. Intersection
and union are implemented as commutative, associative and idempotent
operators, so changing the order of their arguments does not change the result.

There are many further derived rules that can be beneficial in
reducing the state space.  Unions and intersections are both
implemented by sets.  If a union contains a regex $S$, such as a
predicate $\psi$, that is trivially subsumed by another regex $R$,
such as $\psi\st$, then $S$ is removed from the union. This is an
instance of the loop rule in Figure~\ref{fig:rw} that rewrites
$\RELoop{\psi}{0}{\infty}{\alt}\RELoop{\psi}{1}{1}$ to
$\RELoop{\psi}{0}{\infty}$.

A further simplification rule for unions is that if a union contains
a regex $\psi\st$ and all the other alternatives only refer to
elements from $\den{\psi}$ then the union reduces to $\psi\st$.  This
rule rewrites any union such as (${.\st}\texttt{ab}.\st\alt .\st$) to
just $.\st$ (recall that $.\equiv\texttt{[\caret\bslash{n}]}$), which
significantly reduces the number of alternations in total.  Such
rewrite could potentially be extended even further, e.g., rewriting
($.\st\texttt{ab}.\st\alt.\st\texttt{b}.\st$) to $.\st\texttt{b}.\st$,
but the detection overhead is more complex for rewrites like this
and has not been evaluated yet.

\begin{figure}
\[
   \inferrule
   {\rnot(\top\st)}
   {\bot}
\quad
   \inferrule
   {\rnot\bot}
   {\top\st}
\quad
   \inferrule
   {\rnot\rnot R}
   {R}
\quad
   \inferrule
   {\rnot\eps}
   {\top\plus}
\quad
   \inferrule
   {\rnot(\top\plus)}
   {\eps}
\quad
   \inferrule
   {\bot{\cdot}R}
   {\bot}
\quad
   \inferrule
   {R{\cdot}\bot}
   {\bot}
\quad
   \inferrule
   {\eps{\cdot}R}
   {R}
\quad
   \inferrule
   {R{\cdot}\eps}
   {R}
\quad
   \inferrule
   {\bot\st}
   {\eps}
\]
\[
\quad
   \inferrule*
   {\top\st \alt R}
   {\top\st}
\quad
   \inferrule*
   {\top\st \rand R}
   {R}
\quad
   \inferrule*
   {\bot \alt R}
   {R}
\quad
   \inferrule*
   {\bot \rand R}
   {\bot}
\quad
   \inferrule*
   {\RELoop{R}{1}{1}}
   {R}
\quad
   \inferrule*
   {\RELoop{R}{0}{0}}
   {\eps}
\]
\[
   \inferrule*[right=${(\IsNullable[x]{R})}$]
   {\eps\rand R}
   {R}
\quad
   \inferrule*[right=${(\bnot\,\IsNullable[x]{R})}$]
   {\eps\rand R}
   {\bot}
\quad
   \inferrule*[right=${(l\leq k\leq m)}$]
   {\RELoop{R}{l}{m}\alt\RELoop{R}{k}{n}}
   {\RELoop{R}{l}{\max(m,n)}}
\]
\caption{Basic rewrite rules.\label{fig:rw}}
\end{figure}

\subsection{POSIX Semantics}

One key difference from the .NET 7 nonbacktracking engine, is that we
treat both alternations and intersections as sets, which makes the
semantics different from the backtracking one.  For example, matching
the regular expression \texttt{(a|ab)*} on the string
``\texttt{abab}'' with the current .NET engines results in the prefix
``\texttt{a}'', as the leftmost alternation finds a successful match. Our
engine will match the whole string ``\texttt{abab}'', as all the
alternations are explored in parallel, and the longest match is in the
right alternation. The behavior of $\IsMatch{x}{R}$ is identical in both
engines.

In the nonbacktracking engine, PCRE semantics is achieved by
prepending $\top\st\lazy$ to the regex pattern, which is always
nullable and has only two potential derivatives.  For example, in the
case of the pattern $\top\st\lazy\texttt{ab}$, the derivative can be either
$\texttt{b}{\alt}\top\st\lazy\texttt{ab}$, when the character matches $\CCpred{a}$ (the initial
pattern creates a new alternation) , or $\top\st\lazy\texttt{ab}$, (i.e. just the initial
pattern) when it does not.

In our case we meet POSIX semantics by keeping the initial pattern
as is, but starting a potential match on every input position, 
that has a valid starting predicate. Then we keep track of the nullability of all
alternations the initial pattern produced separately. This is done by
having a set of "top-level alternation" data structures, each of which
consist of a regex and the maximum nullable position of the
alternation.  Once an alternation turns into \texttt{$\bot$}, we check if the alternation has a
nullable position, and if it does, we return the match. If the alternation
does not have a nullable position, we remove it from the set and
continue matching the remaining alternations.

\subsection{Conjunctive Startset Lookup and Skipping}
\label{sec:startset-lookup}

One very important optimization that allows to avoid doing a lot of
unnecessary work on average, depending on the input string and regular
expression, is the startset optimization, which enables search for
characters belonging to the set of initial characters of a regex
denoted by the \textit{startset} predicate by using dedicated string
traversal constructs which can be several orders of magnitude faster
than checking characters one-by-one.

The result of evaluating the predicate \textit{startset} is captured in the
state of the regex node. Startset optimization
is often comprised of checking if the first character of the regex
pattern is present in the input string.  As mentioned in Section
\ref{ex:paragraph-matching}, we make use of the fact that determining
the startset of a regex with derivatives is a relatively cheap
operation, and we can use it to incorporate efficient vectorized
string-search procedures into the matching algorithm.

A key difference from other regex engines is that we don't just skip
until the initial startset of the regex, but perform intermediate
startset computations and appropriate skipping in the inner loop as
well.  To take advantage of the fast \textit{startset} check, we first need
to determine if a regex is ``skippable'' or not, where we determine if
taking the derivative changes the node on only a
subset of the input alphabet. Such regexes start with a
\texttt{$\st$}-loop.

A skippable regex is defined as a regex that starts with a *-loop, or one in which
all child regexes recursively start with a *-loop, such as an alternation, lookahead,
intersection or negation of regexes starting only with *-loops. Lookbacks are more complex
and do not currently take advantage of this optimization. The key insight of this 
skippable check is that *-loops keep producing the same regex derivative
until either the loop terminator or tail of one of the *-loops matches.

For example, all regexes starting with $\top\st$ are skippable, as according to the
derivation rules in Section ~\ref{sec:derivatives}:
the derivative of the regex does not change unless the tail of the concatenation produces a
non-$\bot$ derivative. The same is true for all regexes starting with $.\st$, as
the only two predicates that produce a non-initial derivative are
$\CCpred{\bslash{n}}$, that is the loop termination predicate for \texttt{.*}, 
and the starting predicate of the concatenation tail. For example, the regex
\texttt{.\st{}ab} has two predicates that change the state:
$\CCpred{\bslash{}n}$, in which case the regex becomes $\bot$, and
$\CCpred{a}$, in which case the regex becomes \texttt{b|.*ab}. All other
predicates produce the initial regex.

Such skipping becomes useful when we want to match over input text with
multiple intersections or alternations. For example, consider the regex \texttt{(12.*)\&(.*{\textbackslash}d)},
which is strictly non-skippable, as the transition from \texttt{12.*} to \texttt{$\bot$} or to \texttt{2.*} 
is an important part of the matching process. However, the regex \texttt{(.*12.*)\&(.*{\textbackslash}d)} is 
skippable, allowing us to look up the occurrence of a character satisfying the
predicate $\CCpred{\texttt{[{\textbackslash}d{\textbackslash}n]}}$ (the disjunction of 
predicates \texttt{1}, \texttt{{\textbackslash}d} and the loop exit condition \texttt{{\textbackslash}n}),
and continue matching directly at that position.

The in-match startset computation heuristic itself is very simple, and only considers the first node
of the regex, or the first node and the second node recursively in the case of a star. In the 
case of a regex starting with a predicate, we take the predicate itself.
In the case of a predicate star loop, we take the negation of the loop body predicate, 
and union it with the concatenation tail predicate, as done in the previous example. 
In the case of an intersection,
alternation or complement, we take the union of the startsets of the contained regex bodies.
Notably, the startset of a complement is not the complement of the startset of the contained regex, but
the same startset, as the derivation rules apply the same way inside negation.

After computing the joint startset of the regex, we use it in a vectorized index-of
lookup in the following input string. Currently, we perform this lookup in an overestimated manner,
by taking the union of all startsets, but this can be improved upon by taking the intersection in 
some cases. For example, if we have an intersection with the startsets $\CCpred{\bslash{}d}$ and $\CCpred{1}$ 
strictly in the same position, then we know ahead-of-time that the startset of the intersection is constrained 
to only $\CCpred{1}$ - the conjunction of the two startsets.

The startset optimization is very powerful as it allows us to skip over large parts of the input string
in a single step, and only perform the more expensive derivative computation on the remaining
positions. A future improvement would be to extend this to multi-character predicates and intersections, 
which would allow us to skip over even more of the input string.

\section{Evaluation}
\label{sec:evaluation}

The evaluation of the engine is still in its early stages, and we 
have not yet implemented all the optimizations that are possible,
such as caching transition regexes, or skipping over multi-character predicates.
However, despite being only a research prototype, we can already see that the engine 
is capable of solving problems that neither of the available .NET regex engines can 
currently solve. 

An industrial implementation of the engine should reach performance characteristics
comparable to \cite{PLDI2023} over standard regexes, because the same .NET framework is used, 
with potential marginal gains in space/time due to commutative \texttt{|}, as
non-commutativity of \texttt{|} disallows some useful rewrites in \cite{PLDI2023}.
Moreover, all initial fixed-prefix/suffix-search optimizations are shared across all backends. 
For a more comprehensive evaluation on standard regexes outside .NET, see the evaluation of 
\cite{PLDI2023} directly.

Note that we do not use the caching mechanism from the nonbacktracking engine 
in our prototype,
as it is not yet fully implemented, and would require a significant amount of work to integrate in
the presence of lookarounds. The lack of caching is the main reason why our engine is not yet
competitive with the nonbacktracking engine on the benchmarks.

We have evaluated the performance of our engine against the .NET default 
backtracking engine, and the symbolic automata based nonbacktracking version.
The benchmarks were run on a machine with an AMD Ryzen 7 5800X 8-Core CPU, and
32 GB of RAM. The benchmarks were run on .NET version 7.0.305.

We compare the performance of extracting paragraphs containing multiple substrings
from the collected works of Mark Twain \cite{twain}, first we compare the performance on a 
9 kB string, containing 34 paragraphs, extracted from the lines
188589 to 188771, and then on the entire
20 MB file, containing 379897 lines in total.

The paragraphs are extracted with three kinds of patterns: one that uses a negative 
lookahead to match the end of the paragraph, one that uses a line loop until the occurrence
of two sequential newlines, and one that uses negation to constrain the paragraph range 
and intersections to constrain the paragraph contents.

Examples of the patterns used to match paragraphs containing the word "King":
\begin{enumerate}
    \item Negative lookahead: \verb|\n\n((?!\n\n)[\s\S])*?(King)((?!\n\n)[\s\S])*?\n\n|
    \item Loop: \verb|\n\n((.+\n)+?(.*King.*\n)(.+\n)+?)\n|
    \item Conjunction, negation: \verb|\n\n~([\s\S]*\n\n[\s\S]*)\n&[\s\S]*King[\s\S]*|
\end{enumerate}

Note that using a simpler pattern, such as
\verb|\n\n[\s\S]*?King[\s\S]*?\n\n|, would not work, as any non-match inside the
paragraph would cause \verb|[\s\S]*?| to leak outside the current paragraph, 
and match over multiple paragraphs.

The substrings we look for in the paragraph are "King", "Paris", "English", "would", "rise", "struck", "council",
"march", "war", "May", "Orleans" and "work", where we introduce the next substring, as the number of substrings increases. 
In the case of 1 substring we use "King", for 2 substrings we look for "King" and "Paris" in the same paragraph, etc.

\begin{figure}
  \includegraphics[scale=1.0]{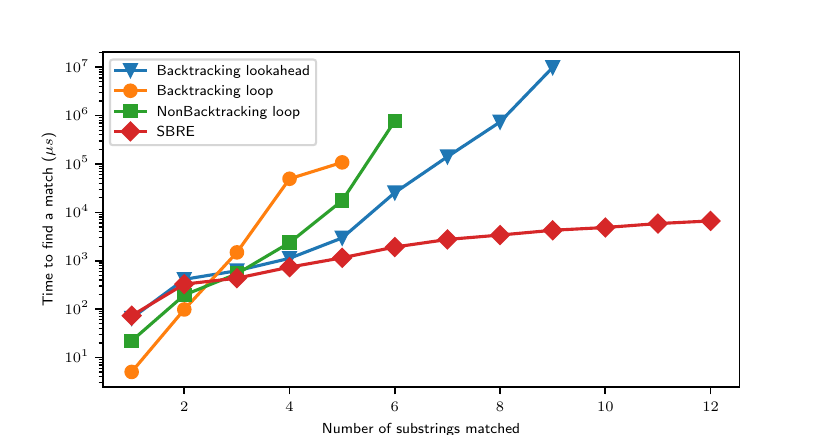}
  \caption{Time taken to find a matching paragraph containing all required substrings in any order in a 9kB sample text.}
  \label{fig:parawords}
\end{figure}

The results of an experiment run on a 9 kB string where the first paragraph matching the number of substrings indicated on the x axis is given in Fig. \ref{fig:parawords}. As the y axis is logarithmic it can be seen how using the intersections produces more efficient match as there is no blowup in match orderings required in conventional regular expressions.

In the 9 kB sample text, our engine, SBRE \cite{sbre_zenodo} starts becoming faster than the rest at 3 substrings, 
but due to lack of caching has 4 times higher memory allocations than the nonbacktracking engine and about 9 times higher 
allocations than the backtracking implementation. 

After a certain number of substrings, the order of the matchings (Table ~\ref{tab:pattern-blowup}) becomes so large that both
the backtracking and nonbacktracking engines stop being able to find a match within 60 seconds, while our engine still manages
to complete the match with minor slowdowns.

\begin{table}
    \caption{Paragraph extraction pattern lengths in characters}
    \label{tab:pattern-blowup}
    \begin{tabular}{p{1.0in}p{1.3in}p{1.0in}p{1.0in}}
      n of substrings & Neg. lookahead & Loop & \& and $\sim$ \\
      \hline
      1 & 50    & 34        & 46 \\
      2 & 101   & 77        & 66 \\
      3 & 363   &  295      & 88 \\
      4 & 1869  & 1513      & 108 \\
      5 & 11805 & 9385      & 127 \\
      6 & 87885 & 69145     & 148 \\
      7 & 740925 & 579625   & 170 \\
      8 & 6854445 & 5322265     & 190 \\
      9 & 69310125 & 53343385   & 208 \\
      10 & 769305645 & 587865625    & 226 \\
      \hline
    \end{tabular}
\end{table}

In the full 20 MB text, our research prototype (SBRE) takes significantly more unnecessary 
derivatives than the nonbacktracking variant, as it suffers from the lack of 
caching,
but after a certain number of substrings, it is still the only engine that can complete the match
within 60 seconds. The results are shown in Fig. \ref{fig:fulltextwords}.

A notable observation in the full text is that the lookahead pattern of the backtracking engine falls
off significantly earlier than in the 9 kB text, where it showed promising results up to
6 substrings. In the full text, the nonbacktracking engine shows the best performance up to 6 substrings,
but then falls off completely, as the memory allocated blows up from 74.41 MB in 5 paragraphs to 
3.39 GB in 6 paragraphs. The nonbacktracking engine was not able to produce a result with 
7 substrings within 10 minutes.

\begin{figure}
    \includegraphics[scale=1.0]{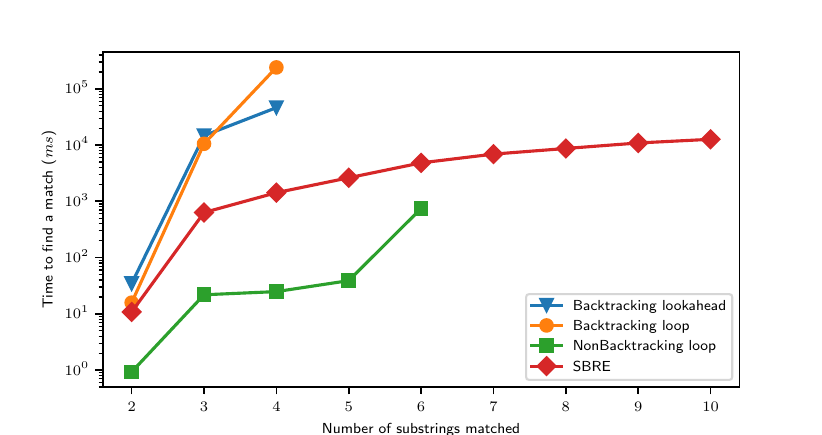}
    \caption{Time taken to find a matching paragraph containing all required substrings in any order in the full 20MB sample text.}
    \label{fig:fulltextwords}
\end{figure}
  
The results show that our engine is able to search for a large amount of substrings in a paragraph
efficiently, while the other engines suffer from the exponential blowup in match orderings.
In the full text, adding a substring to the SBRE pattern after 2 substrings currently causes a constant memory 
allocation growth of 800 MB, which is a significant amount of memory, but still allows for the match to complete,
while the other engines are not able to complete the match at all. The memory usage of the SBRE
engine could be significantly reduced by caching the derivatives, which would make the
engine more competitive overall.

\section{Related Work}
\label{sec:related}
Regular expressions have in practice many extensions, such as
\emph{backreferences} and \emph{balancing groups}, that reach far
beyond \emph{regular} languages in their expressive power. Such
extensions, see~\cite{LMK19}, fall outside the scope of this
paper. The focus on related work here is solely on automata and
derivative based matching algorithms, tools, and techniques, that in
some form or shape maintain, at least in principle, a finite state
based view corresponding to regular languages.  In particular,
\emph{lookaheads} maintain regularity~\cite{Morihata12} and regular
expressions with lookaheads can be converted to Boolean
automata~\cite{Berglund21}.  \cite{DBLP:journals/ieicetd/ChidaT23}
consider extended regular expressions in the context of backreferences
and lookaheads. They build on \cite{DBLP:conf/lata/CarleN09} to show
that extended regular expressions involving backreferences and both
positive and negative lookaheads leads to \emph{undecidable}
emptiness, but, when restricted to positive lookaheads only is closed
under complement and intersection.  \cite{Miyazaki2019} present an
approach to find match end with derivatives in regular expressions
with lookaheads. The semantics of derivatives in~\cite{Miyazaki2019}
uses \emph{Kleene algebras with lookahead} as an extension of Kleene
algebras with tests~\cite{Kozen97}, and is fundamentally different
from our formulation in several aspects, where the underlying semantic
concatenation is \emph{commutative} and \emph{idempotent} and where
the difference to derivatives of concatenations in relation
to~\cite{Brz64} is also pointed out, which also leaves unclear the
question of supporting lookbacks and reverse.  Derivatives in
combination of Kleene algebras are also studied in~\cite{Pous15}.  Our
approach is a conservative extension of the theory in~\cite{PLDI2023}
with intersection and complement, as well as positive and negative
lookbacks and lookaheads, that builds on~\cite{Brz64}, leading to the
core fundamental results of Theorem~\ref{thm:DERS} and
Theorem~\ref{thm:REV}.  

In functional programming derivatives were studied
in~\cite{Fis10,Owens09} for {\IsMatchName}.  \cite{SuLu12,ADU16} study
{\FindMatchEndName} with Antimirov derivatives and POSIX semantics and
also Brzozowski derivatives in \cite{ADU16} with a formalization in
Isabelle/HOL. The algorithm of \cite{SuLu12} has been recently further
studied in~\cite{TanUrban23} as a recursive functional program and
also formalized in Isabelle/HOL.  The fundamental difference to the
theory here is that, because of lookarounds (that in particular enable
the definition anchors) imply that certain classical laws, such as several of the
\emph{inhabitation relation rules} in~\cite{TanUrban23}, become
invalid.  In~\cite{Wing19} anchors are considered as special symbols
in an extended alphabet, using classical derivatives.  This approach
has the drawback that anchors are intended as specialized lookarounds
and treating them as special symbols conflicts with match semantics in
terms of locations.  Some aspects of our work here, such as support
for intersection, are related to SRM~\cite{Vea19} that is the
predecessor of the \textsc{NonBacktracking} regex backend of
{.NET}~\cite{PLDI2023}, but SRM lacks support for lookarounds as well
as anchors. The top-level matcher of SRM is also different and more
costly because it uses three passes over the input to locate a match,
instead of two.

State-of-the-art nonbacktracking regular expression matchers based on
automata such as RE2~\cite{Cox10} and grep~\cite{grep} using variants
of~\cite{Thom68}, and Hyperscan~\cite{HyperscanUsenix19} using a variant
of~\cite{Glu61}, as well as the derivative based
\textsc{NonBacktracking} engine in {.NET} make heavy use of
\emph{state graph memoization}.  None of these engines currently
support lookarounds intersection or complement.  An advantage of using
derivatives is that they often minimize the state graph (but do not guarantee
minimization), as was
already shown in~\cite[Table~1]{Owens09} for DFAs. Further evidence of
this is also provided in~\cite[Section~5.4]{SuLu12} where NFA sizes
are compared for Thompson's and Glushkov's, versus Antimirov's
constructions, showing that Antimirov's construction consistently
yields a smaller state graph.  In automata-based engines an upfront
DFA minimization is undesirable because it is too costly, while
derivatives allow DFA-minimizing optimizations to be applied
essentially on-the-fly.  In general, the rewrite rules applied in our
framework are not feasible in automata-caching, because the
corresponding semantic language-level checks would require global analysis, 
because in traditional automata based representations one has lost the 
relationship between states and regular expressions. In our case we preserve
 the relationship between regular expressions

The two phases of the top-level matching algorithm: to find the match
end location and to find the match start location are similar in
RE2~\cite{Cox10} as well as in \textsc{.NET NonBacktracking} which our
implementation builds on top of.  It therefore also benefits from
switching to NFA mode when a certain threshold of DFA states is
reached, having an overall effect similar to that of RE2~\cite{Cox10},
but the derivatives can switch to Antimirov-style derivatives without
any prior bookkeeping. The top-level loop is in some sense oblivious
to the fact that intersection, complement, and lookarounds are being
used inside the regular expressions.

The two main standards for matching are PCRE (backtracking semantics)
and POSIX~\cite{Lau2000,Berg21}.  \emph{Greedy} matching algorithm for
backtracking semantics was originally introduced in~\cite{FC04}, based
on $\epsilon$-NFAs, while maintaining matches for eager loops. We
should point out that \cite[Proposition~2]{FC04} assumes the axiom
$L(R{\cdot}S)=L(R){\cdot}L(S)$ that fails with anchors or lookarounds.
While backtracking semantics is needed in {.NET} for compatibility
across all the backends -- including
\textsc{NonBacktracking}~\cite{PLDI2023} -- here we use POSIX
semantics that allows us to treat alternations and intersections as
commutative operations. The rationale behind using POSIX is that it is
semantically unclear as to what de Morgan's laws and laws of
distributivity would mean in the context of backtracking semantics if
intersection would be treated as a noncommutative operation.

The results \cite[Theorem~3.3 and Theorem~3.8]{PLDI2023} form a
\emph{strict subset} of our Theorem~\ref{thm:DERS} and
Theorem~\ref{thm:REV} whose proofs are novel and nonobvious because definitions of
nullability and derivatives are mutually recursive. It was far from
obvious if regex complement and intersection could even be combined in
any meaningful way with lookarounds.  It remains unclear to us, as to
if the main result~\cite[Theorem~4.5]{PLDI2023} that builds on
formally linking the semantics of derivatives to \emph{backtracking}
(PCRE) semantics can be extended to cover the extended fragment of
regexes defined here by $\RE$ due to noncommutativity of alternations
in backtracking semantics, the proof of~\cite[Theorem~4.5]{PLDI2023}
is much more complicated than that of Theorem~\ref{thm:Match}.
Intersection was also included as an
experimental feature in the initial version of SRM~\cite{Vea19} by
building directly on derivatives in~\cite{Brz64}, and used an encoding
via regular expression \emph{conditionals} (if-then-else) that
unfortunately conflicts with the intended semantics of conditionals
and therefore has, to the best of our knowledge, never been used or
evaluated.

The conciseness of using intersection and complement in regular
expressions is also demonstrated in~\cite{Gelade2012}
where the authors show that using intersection and complement in
regular expressions can lead to a double exponentially more succinct 
representation of regular expressions.

\section{Future Work}
\label{sec:future}

The theory of derivatives based on locations that is developed here
can be used to extend regular expressions with lookarounds in SMT
solvers that support derivative based lazy exploration of regular
expressions as part of the sequence theory, such solvers are
CVC4~\cite{CVC4,CVC4deriv} and Z3~\cite{BM08,SVB21}.  A further
extension is to lift the definition of location derivatives to a fully
\emph{symbolic} form as is done with \emph{transition regexes} in
Z3~\cite{SVB21}. \cite{10.1145/3498707} mention that the OSTRICH
string constraint solver could be extended with backreferences and
lookaheads by some form of alternating variants of prioritized
streaming string transducers (PSSTs), but it has, to our knowledge,
not been done.  Such extensions would widen the scope of analysis of
string verification problems that arise from applications that involve
regexes using anchors and lookarounds.  It would then also be
beneficial to extend the SMT-LIB~\cite{SMTLIB} format to support
lookarounds.

Counters are a well-known Achilles heel
of essentially all nonbacktracking state-of-the-art regular expression
matching engines as recently also demonstrated in~\cite{THHLVV22}, which makes any
algorithmic improvements of handling counters highly valuable.
In~\cite{CsA20}, Antimirov-style derivatives~\cite{Ant95} are used to
extend NFAs with counting to provide a more succinct symbolic
representation of states by grouping states that have similar behavior
for different stages of counter values together using a data-structure
called a \emph{counting-set}.  It is an intriguing open problem to
investigate if this technique can be adapted to work with location
derivatives within our current framework.  
\cite{DBLP:journals/pacmpl/GlaunecKM23} point out that it is important
to optimize specific steps of regular expression matching to address
particular performance bottlenecks. The specific BVA-Scan algorithm is
aimed at finding matches with regular expressions containing counters
more efficient. \cite{DBLP:conf/fossacs/HolikSTV23} report on a subset
of regexes with counters called synchronizing regexes that allow for fast
matching.

Further work of improved rewrite rules is also needed in our current
implementation to reduce the number of states that arise in the
matching engine, as well as enhancing caching techniques that record
transitions that have already been seen for lookarounds. Currently, no
special handling of transitions arising from lookarounds is taking
place. And, as intersections allow precise text extraction in a single
pass, the approach can be optimized further by application of
vectorization.

One of the main drawbacks of lookarounds is that they are not
cacheable in the same way as other regexes, as they depend on the
position of the matcher in the string. This means that the regex
engine cannot statically determine if a lookaround will match at a
given position and has to try the lookaround at every position.

This is not a problem for the short context-lookups, like the
anchor lookarounds in Table \ref{tab:anchors}, as
they do not move around in the string, and retain the linear time
complexity of the regex. However, this is a problem for the other
lookarounds, such as the lookback \texttt{(?<=a.*)b}, where the
occurrence of the character \texttt{a} has to be re-checked at every
occurrence of \texttt{b}.  This leads to a non-ideal quadratic time complexity.

However, these kinds of lookarounds could be cached by storing the
predicate $\DOTpred=\CCpred{\texttt{[\caret\bslash{n}]}}$ and
passing it along with the derivative. This way, the regex engine will
only invalidate the cache upon finding an occurrence of a character
that does not match the loop predicate, i.e., in the case of \texttt{(?<=a.*)}
this cached predicate will be invalidated by $\bslash{n}$, which is the only
character not in $\den{\DOTpred}$. This caching technique could lead
to a linear time complexity for many lookarounds but needs more
research.

One notable feature that is missing from our regexes is support for
laziness. However, this feature can often be emulated by using
complement and negation of character classes. For example, the regex
\texttt{a.*?b}, which matches the shortest string in the range of
\texttt{a} and \texttt{b}, can be emulated by
\texttt{a[\caret{}b\bslash{n}]*b}, which has the exact same semantics,
but without using laziness. Another way to emulate laziness is by
using intersection and complement, as in
\texttt{(a.*)\&(\rnot(.*b.*)b)}, which is also equivalent
\texttt{a.*?b}.  Note that \texttt{\rnot(.*b.*)}, that means ``does
not contain \texttt{b}'', matches \emph{before} \texttt{b} where the
final match character must be \texttt{b} in \texttt{(\rnot(.*b.*)b)}.

\section{Conclusion}
We have presented both a theory and implementation for a combination
of extensions to regular expressions including complement,
intersection and positive and negative lookarounds that have not
previously been explored in depth in such a combination.  Prior work
has analyzed different other sets of extensions and their properties,
but several such combinations veer out of the scope of regular
languages. The work combines and extends the symbolic derivatives
based approaches presented in \cite{Vea19} and \cite{PLDI2023} by
showing how the carefully selected set of extensions to regular
expressions yield an effective Boolean algebra on the match set
semantics, producing a regular language with interesting new
applications and opening up a principled way to defining rewrite rules
for optimizations that play an important role in practical
applications. In addition, we have provided a precisely defined
approach to finding matches using symbolic derivatives.

When considering context-sensitive anchors, we showed how anchors can
be represented in terms of lookarounds and are thus supported by our
approach. Moreover, such generalization provides possibilities for
defining custom anchors.

The implementation reuses several components of the .NET 7
nonbacktracking regular expression backend and adds support for newly
introduced extensions. To demonstrate the efficacy of the approach we
showed that the task of locating a substring containing a set of
matches in arbitrary order can be solved by the proposed engine while
neither the nonbacktracking nor backtracking engines of .NET 7 scaled
due to the factorial blowup of the possible orderings of the matches.
 
\bibliographystyle{ACM-Reference-Format}
\bibliography{bib}

\end{document}